\documentclass[twocolumn,letter]{IEEEtran}
\usepackage{cite}

\ifCLASSINFOpdf
   \usepackage[pdftex]{graphicx}
  \DeclareGraphicsExtensions{.pdf,.jpeg,.png}
\else
  \usepackage[dvips]{graphicx}
  \DeclareGraphicsExtensions{.eps}
\fi
\usepackage[cmex10]{amsmath}
\usepackage{algorithmic}

\usepackage{array}
\usepackage{fixltx2e}

\ifCLASSOPTIONcaptionsoff
  \usepackage[nomarkers]{endfloat}
 \let\MYoriglatexcaption\caption
 \renewcommand{\caption}[2][\relax]{\MYoriglatexcaption[#2]{#2}}
\fi
\usepackage{url}

\usepackage{epsfig}
\usepackage{multicol}
\usepackage{verbatim}
\usepackage{amssymb}
\usepackage{epstopdf}
\usepackage{amsthm}
\usepackage{amsfonts}%
\usepackage{enumerate}%
\usepackage{psfrag}
\usepackage{setspace}
\usepackage{dsfont}
\usepackage{balance}

\usepackage{algorithm}
\usepackage{float}

\usepackage[section]{placeins}
\usepackage{xcolor}

\newcommand\blfootnote[1]{%
  \begingroup
  \renewcommand\thefootnote{}\footnote{#1}%
  \addtocounter{footnote}{-1}%
  \endgroup}

\newtheorem{proposition}{Proposition}

\newtheorem{remark}{Remark}

\begin{document}

\title{On Orthogonal Band Allocation for Multi-User Multi-Band Cognitive Radio Networks: Stability Analysis}
\author{ Ahmed El Shafie$^\dagger$, Tamer Khattab$^*$\\
\small \begin{tabular}{c}
$^\dagger$Wireless Intelligent Networks Center (WINC), Nile University, Giza, Egypt. \\
$^*$Electrical Engineering, Qatar University, Doha, Qatar. \\
\end{tabular}
}
\date{}
\maketitle
\thispagestyle{empty}
\pagestyle{empty}
\blfootnote{Part of this paper was presented in IEEE Wireless Communications and Networking Conference (WCNC), 2014 \cite{ElSh1404:Band}.}
\blfootnote{This paper was made possible by a NPRP grant 6-1326-2-532 from the
Qatar National Research Fund (a member of The Qatar Foundation). The
statements made herein are solely the responsibility of the authors.}

\begin{abstract}
 In this work, we study the problem of band allocation of $M_s$ buffered secondary users (SUs) to $M_p$ primary bands licensed to (owned by) $M_p$ buffered primary users (PUs). The bands are assigned to SUs in an orthogonal (one-to-one) fashion such that neither band sharing nor multi-band allocations are permitted. In order to study the stability region of the secondary network, the optimization problem used to obtain the stability region's envelope (closure) is established and is shown to be a linear program which can be solved efficiently and reliably. We compare our orthogonal allocation system with two typical low-complexity and intuitive band allocation systems. In one system, each cognitive user chooses a band randomly in each time slot with some assignment probability designed such that the system maintained stable, while in the other system fixed (deterministic) band assignment is adopted throughout the lifetime of the network. We derive the stability regions of these two systems. We prove mathematically, as well as through numerical results, the advantages of our proposed orthogonal system over the other two systems.

\end{abstract}

\section{Introduction}
\IEEEPARstart{T}\small{here} is a recent dramatic increase in the demand for radio spectrum stimulated by the enormous influx of new wireless devices and applications. The cognitive radio communications paradigm enables efficient use of the electromagnetic spectrum. Cognitive or secondary users utilize the spectrum when it is unused by the primary or licensed user. In a typical real--life scenario, such as a secondary network of wireless sensors tapping into spectrum holes of a primary cellular network, multiple cognitive users are trying to utilize spectrum holes in a primary multi-band network. In these scenarios, the design of an efficient spectrum allocation protocol to assign the secondary users (SUs) to the available primary bands is very crucial.

The problem of band allocation in cognitive radio networks has been studied in different settings within the literature~\cite{liu2008distributed,liu2008randomized,digham2008joint,lu2009optimal,gai2010learning,lai2011cognitive,shiang2008queuing,queues}.  In order to avoid convergence to the same channels, the authors in~\cite{liu2008distributed} propose a simple distributed sensing policy where each SU individually decides on a single channel to sense, at every time slot, with the objective of maximizing the probability of finding the channel idle while minimizing the probability of colliding with other SUs. A suboptimal randomized channel access policy is derived. The channel access probability for each SU is determined by its belief, which is the conditional probability given all past decisions and observations, that the channels are in a particular state of occupancy by the primary users (PUs). The system is a type of history-based greedy method, which cannot guarantee the optimality of the solution.  Moreover, the system assumes a Markov based model for channel occupancy state, which is not necessarily the case in all systems. In \cite{liu2008randomized}, the system model is changed to assume that each user can sense multiple channels at the same time. The probabilities of sensing the different channels are assigned to the SUs, and the sensing policy is formulated as an optimization problem over all combinations of the assignment probabilities to maximize the total throughput of the network.  While the work addresses the optimal strategy for multi-user multi-band cognitive allocation, it ignores the existence of buffers (queues) in primary and SUs, which is typically the practical case.  Furthermore, the effect of time slots wasted by the SUs to perform channel sensing is not taken into consideration.  In addition, the practicality of sensing multiple channels at the same time is questionable, since it mandates a transceiver capable of aggregating multiple bands at the same time while dealing with each one independently, which requires multiple radio frequency (RF) chains. Moreover, channel fading and noise effects are not considered in the studied system.

 The work in~\cite{digham2008joint} investigates the case where a set of channels is distributed among multiple secondary nodes that opportunistically access the available spectrum with optimal power allocation. The solution of the band allocation problem is obtained via maximizing the total sum capacity of the cognitive radio network both with and without interference constraints on the PUs.  The solution is found to be a modified form of water filling.  By introducing an interference
temperature constraint to guarantee PUs' quality of service (QoS), the authors of \cite{lu2009optimal} proposed an optimal subcarrier and
power allocation algorithm to maximize the overall utility
for SUs. The authors of \cite{gai2010learning} considered the optimal matching of SUs to primary channels in a stochastic setting
as a combinatorial multi-armed bandit problem. Each of the SUs selects a channel to sense and access according to some
policy. The objective is to find an allocation of channels for all SUs that maximizes the expected sum throughput. They investigated a naive policy that ignores the dependencies between the arms and developed a sophisticated policy
that matches learning with polynomial storage. In \cite{lai2011cognitive}, a cognitive medium access protocol is proposed for uncertain environments where the PU traffic statistics are unknown a priori and have to be learned and tracked. In the case of multiple SUs, the channel selection is formulated as an optimization problem for cooperative SUs and a non-cooperative game for selfish SUs, respectively. The presence of data queues as well as the effect of non-negligible sensing time in the system has not been considered in all the aforementioned work.

Resource allocation involving buffer dynamics in a cognitive setting has been considered in a few works such as \cite{shiang2008queuing} and \cite{queues}. In \cite{shiang2008queuing}, a dynamic channel-selection for autonomous wireless users is proposed, where each user has a set of actions and strategies. Based on priority queueing analysis (i.e., priority classes among SUs), each wireless user can evaluate its utility impact based on the behavior of
the users deploying the same frequency channel including the
PUs. The work in \cite{queues} investigates the resource allocation problem for the downlink of an orthogonal frequency division multiple access (OFDMA) based cognitive radio network. Prior to the beginning of each frame, each user transmits to the base station its sensing
information vector as well as its latest channel gain vector, which was obtained based on pilot symbols. Based on the received information from the users and the current backlog for each user, the base station performs resource allocation for the frame. The resource allocation map is then sent to the users and is
valid for the remainder of the frame, which is composed of multiple time slots. The aforementioned work uses a utility based approach to achieve a certain QoS requirement for the SUs.  However, it doesn't address the fundamental limits on performance under the assumption of buffered users in different channel allocation schemes, which is one of the main contributions of our work.


In this work, we propose a novel orthogonal channel allocation scheme for cognitive users.  We study the throughput closure of our proposed scheme as well as revisit schemes previously proposed in the literature considering buffered users, time slotted channels and include channel outage effects on the system's performance. We do not assume the availability of channel state information (CSI) at the transmitting terminals. In our proposed system, we consider a time-slotted primary channel over which each PU starts transmitting at the beginning of the time slot whenever it has packets to communicate. Each PU uses a separate band (sub-channel) of the channel with a certain bandwidth. The permutations of the SUs orthogonal assignment (a single user is assigned exclusively to a single band) to the different bands are probabilistically generated at the beginning of each time slot. Each SU senses the primary band assigned to it to detect the activity of the PU owning the band and will only transmit in case the PU is idle. By varying the assignment permutation probabilities, we can obtain the maximum stable-throughput region for the secondary network.
 To the best of our knowledge, the investigation of the considered systems from the network layer standpoint is addressed in this paper for the first time.  The following is a list of what we believe are the new contributions in this paper:
\begin{itemize}
\item We propose a novel orthogonal channel allocation scheme for cognitive radio networks composed of multiple PUs and SUs.
\item We study the stability region of the proposed system as well as two reference low-complexity intuitive channel allocation systems that have been previously proposed in the literature; namely {\it random selection of bands} and {\it fixed band allocation}, comparing their performance to the proposed system while taking buffers and channel outages into account without CSI at the transmitter side.
\item We are able to mathematically model the throughput closure of our proposed system by constructing an exact optimization model for its maximum stable-throughout region, which is  shown to be a linear program. Then, we provide several important exact solutions for the stability regions and the assignments policy in our system for the important example cases of two SUs and two primary bands, multiple SUs and one primary band, symmetric primary bands, symmetric SUs, and symmetric primary bands with symmetric SUs.
\item We provide mathematical and numerical proofs for the advantages of our proposed system, in terms of throughput closure, over the two classical systems of fixed channel assignments and random selection of bands.
\end{itemize}


\section{system model}
\label{system_model}
We propose a cognitive radio system, denoted by $\mathcal{S}$, in which $M_s$ SUs are assigned to $M_p$ licensed orthogonal
frequency bands. All users operate in a time-slotted fashion. The primary band $j$ has bandwidth $W_{j}$, where in general $W_j\ne W_i$ for $j\ne i$ and $j,i\in \{1,2,\dots,M_p\}$. The secondary network consists of a finite number, $M_s$, of terminals numbered
$1,2,\dots,M_s$. Each terminal, whether primary or secondary, has an infinite queue for storing fixed-length packets \cite{sadek,tvtjournal}.\footnote{We can consider the case of finite queues. However, we will replace the use of Loynes theorem with the constraint that the probability of each of the queues being empty is greater than zero. The characterization of stability will not be possible as we cannot get the closure of rates. Moreover, the constraints will be non-linear; hence, the optimization problem will become a non-convex program. To render the characterization of the stability region tractable, we make use of the widely used assumption of infinite-length queues \cite{sadek,simeone,erph,krikidis2010stability,rao1988stability,krikidis2012stability}. Note that this assumption is a reasonable approximation when the packet size is much smaller than the buffer size \cite{krikidis2012stability}.} The $j$th PU, $p_j$, has a queue denoted by $Q_{{p}_{j}}$, whereas the $k$th SU, $s_k$, has a queue denoted by $Q_{s_k}$. We adopt a discrete-time late-arrival model, which means that a newly arriving packet during a particular time slot cannot be transmitted during the slot itself even if the queue is empty. This model is widely used for queueing systems and has been considered in many papers such as \cite{sadek,tvtjournal,krikidis2012stability} and the references therein. Arrival processes at all queues are Bernoulli random variables that are independent across terminals and independent from slot to slot \cite{sadek,tvtjournal}. The mean arrival rate at $Q_{p_j}$ is $\lambda_{p_j}$ and at $Q_{s_k}$ is $\lambda_{s_k}$. If a terminal transmits during a time slot, it sends exactly one packet to its receiver.

A PU, $p_j$, owning the band $B_j$ (or band $j$), transmits the packet at the head of its queue starting from the beginning of the time slot. The SUs access the channel as follows. Each SU senses the channel assigned to it for a duration of $\tau$ seconds, which is assumed to be a fraction of the slot duration, $T$. We assume that $\tau$ is chosen such that the probability of an erroneous secondary decision regarding primary activity is negligible. If the band is sensed to be free from primary activity, the SU, which is assigned to this band, transmits till the end of the time slot. Note that the transmission time is $T-\tau$ not $T$, but it still transmits one full packet. This can be implemented by the terminal via adjusting its transmission rate, e.g., by using a signal constellation with more symbols or by increasing the channel coding rate or both. Note that by doing this, the probability of link outage increases. This is the price of transmission delay relative to the beginning of the time slot and it is exactly quantified at the end of this section.

\subsection{The Proposed Orthogonal Band Allocation System}
For system $\mathcal{S}$, each band has at most one SU, and each SU is assigned to exactly one band.  We call this system, {\it orthogonal band allocation}. In order to unify the presentation of the orthogonal band allocation method, if the number of SUs is greater than the available primary bands, and since our protocol does not allow multiple assignment of users to the same band, we can assume the presence of $M_s\!-\!M_p$ virtual bands with zero bandwidth. Thus, the service rate on any of these bands is exactly equal to zero. The pattern of the orthogonal allocation of bands to SUs at any time slot is represented by the permutation $\Pi_n$ given by the $M_s$-tuple $(m_1,m_2, \dots,m_k, \dots, m_{M_s})$ over the set of primary bands, $\mathcal{B}$, where
\begin{equation}
\begin{split}
\mathcal{B} = \begin{cases}
 \{0, 1,2,\dots,M_p\} & M_s > M_p\\
\{1,2,\dots,M_p\} & M_s \leq M_p
\end{cases},
\end{split}
\end{equation}
$m_k\in \mathcal{B}$ and $m_k \neq m_\ell \;,\; \forall k \neq \ell$ unless $m_k\!=\!m_\ell\!=\!0$.  The permutation $\Pi_n = (m_1,m_2, \dots, m_{M_s})$ represents the orthogonal assignment pattern that assigns band $B_{m_1}$ to SU $s_1$, band $B_{m_2}$ to SU $s_2$ and so on, with $m_k= 0$ meaning that SU, $s_k$, is assigned to a virtual band with zero bandwidth. At the beginning of each time slot, a predefined SU controller,\footnote{The proposed centralized method can be useful for cognitive radio scenarios where the SUs belong to a heterogenous network such as a wireless sensor network or a secondary cellular network. For similar centralized works, the reader is referred to \cite{digham2008joint} and the references therein.} which can be one of the SUs, randomly generates one of the possible permutations $\Pi_n$ (band assignment pattern) with probability $q(\Pi_n)$. Consequently, each SU knows its allocated band and starts the sensing process independently. It is evident that the assignments are the permutation without repetition of choosing $M_s$ elements out of $M_p$ elements, if $M_p \ge M_s$, or choosing $M_p$ elements out of $M_s$ elements, if $M_s \ge M_p$. Hence, calling the set of all possible band assignment permutations $\mathcal{M}$, the cardinality of $\mathcal{M}$, denoted as $|\mathcal{M}|$, is given by
\begin{equation}
\begin{split}
\label{eqn:cardinality}
|\mathcal{M}| = \begin{cases}
\frac{M_p!}{(M_p-M_s)!} & M_p \ge M_s\\
\frac{M_s!}{(M_s-M_p)!} & M_p < M_s
\end{cases},
\end{split}
\end{equation}
where $r!$ denotes the factorial of $r$. It is clear that the summation over the permutations probabilities is given by
\begin{equation}
\sum_{n=1}^{|\mathcal{M}|}q(\Pi_n) = 1
\end{equation}

Instead of looking at the probability distribution of the different assignment permutations, one can look at a different quantity that deals with the individual assignments of a particular band to a particular user. Let $\omega_{jk}$ denote the fraction of time slots during the lifetime of the network that the SU, $s_k$, is assigned to the band $B_j$. It is evident that the following two constraints on $\omega_{jk}$ must hold:
\begin{equation}
\begin{split}
\sum_{k\!=\!1}^{M_s}\omega_{jk}\le 1, \forall j\in \mathcal{B},
\end{split}
\label{constraints_omega1}
\end{equation}
where equality holds in the case $M_s \ge M_p$; and
\begin{equation}
\begin{split}
\sum_{j=1}^{M_p}\omega_{jk}\le1, \forall k\in\{1,\dots,M_s\},
\end{split}
\label{constraints_omega}
\end{equation}
where equality holds in the case $M_p \ge M_s$. Hence, both constraints become equalities if and only if $M_p\!=\!M_s$. Defining the subset of all possible permutations of band allocations conditioned that band $B_{j}$ is assigned to SU $s_k$ as $\mathcal{M}_{jk} \subset \mathcal{M}$, the relationship among $\omega_{{j}k}$ and $q(\Pi_n)$ can be stated as follows:
\begin{equation}
\begin{split}
\omega_{{j}k}\!=\!\sum_{\Pi_n \in \mathcal{M}_{jk}} q(\Pi_n), \forall k\in\{1,\dots,M_s\}.
\end{split}
\label{constraints_omega2}
\end{equation}

 The probability that band $B_j$ is free/available is the probability that the primary queue assigned to the band is empty. If the queue of user $p_j$ is stable, i.e., $\mu_{p_j}\ge \lambda_{p_j}$, the probability that the queue is empty is given by\footnote{This formula follows from solving the Markov chain of the primary queue under the late-arrival model.}
\begin{equation}
\pi_{j}\!=\!1-\frac{\lambda_{p_j}}{\mu_{p_j}}
\label{eqn1xoo}
\end{equation}
where $\mu_{p_j}$ is the mean service rate of $p_j$ and is given by the complement of the outage event of the channel between the primary transmitter $p_j$ and its respective receiver under perfect sensing assumption. If the queue is unstable, i.e., $\mu_{p_j}< \lambda_{p_j}$, the primary queue is saturated and the probability of the band being available for the SUs is zero. That is, $\pi_{j}=0$ when $\mu_{p_j}< \lambda_{p_j}$. Combining both cases, the probability of the $j$th primary band being available is given by
\begin{equation}
\pi_{j}\!=\!1-\min\{\frac{\lambda_{p_j}}{\mu_{p_j}},1\}
\label{eqn1}
\end{equation}

A feedback message from the respective receiver is sent at the end of each time slot to inform the corresponding transmitter about the decodability status of the transmitted packet.
 If the respective destination decodes the packet successfully, it sends
back an acknowledgement (ACK), and the packet is removed from the transmitter's queue.
 If the respective destination fails to decode the packet due to channel outages, it sends back a negative-acknowledgement (NACK), and the packet is retransmitted at the following time slot.

We summarize MAC algorithm of system $\mathcal{S}$ as shown in Algorithm~\ref{alg:S}.
\begin{algorithm} {}
\caption{$\mathcal{S}$--MAC}
\begin{algorithmic} \label{alg:S}

\WHILE{TRUE}
\STATE{{\it Assignment}:}
\STATE{generate $\Pi_n$ w.p. $q(\Pi_n)$}
\FOR{$\forall s_k$}
\STATE{assign band $B_{j}$ to user $s_k$}
\ENDFOR

\STATE

\STATE{{\it Primary}:}
\FOR{$\forall p_j$}
\IF{$Q_{p_j}$ not empty}
\STATE{transmit packet at head of $Q_{p_j}$}
\ENDIF
\ENDFOR

\FOR{$\forall p_j$}
\IF{ACK received}
\STATE{remove packet at head of $Q_{p_j}$}
\ENDIF
\ENDFOR

\STATE

\STATE{{\it Secondary}:}
\FOR{$\forall s_k$}
\IF{$Q_{s_k}$ is not empty}
\STATE{sense $B_{j}$ for duration $\tau$}
\IF{$B_{j}$ idle}
\STATE{transmit packet at head of $Q_{s_k}$}
\ENDIF
\ENDIF
\ENDFOR

\FOR{$\forall s_k$}
\IF{ACK received}
\STATE{remove packet at head of $Q_{s_k}$}
\ENDIF
\ENDFOR

\ENDWHILE

\end{algorithmic}
\end{algorithm}


We adopt a flat fading channel model and assume that the channel gains remain constant over the duration of the time slot. We do not assume the availability of the CSI at the transmitting terminals.  Assuming that the number of bits in a packet is $b$, the transmission rate of the secondary transmitter $s_k$ is
\begin{equation}
r_{s_k}\!=\!\frac{b}{T-\tau}
\label{r_i}
\end{equation}
\noindent Outage occurs when the transmission rate exceeds the channel capacity \cite{sadek,tvtjournal}
\begin{equation}
{\rm Pr}\bigg\{O_{j,s_k}\bigg\}\!=\!P_{{\rm out},js_k}\!=\!{\rm Pr}\bigg\{r_{s_k} > W_j \log_{2}\left(1\!+\!\gamma_{s_k} \alpha_{js_k}\right)\bigg\}
\end{equation}
\noindent where $O_{j,s_k}$ is the event of channel outage when band $B_j$ is assigned to user $s_k$, $W_j$ is the bandwidth of $B_j$, $\gamma_{s_k}$ is the received signal-to-noise-ratio (SNR) at the receiver of user $s_k$ when the channel gain is equal to unity, and $\alpha_{js_k}$ is the channel gain when user $s_k$ is assigned to band $B_j$, which is exponentially distributed in the case of Rayleigh fading. The outage probability can be written as \cite{sadek,tvtjournal}
\begin{equation}
P_{{\rm out},j{s_k}}\!=\!{\rm Pr}\Bigg\{\alpha_{js_k}<\frac{2^{\frac{r_{s_k}}{W_j}}-1}{\gamma_{s_k}}\Bigg\}
\end{equation}
\noindent Assuming that the mean value of $\alpha_{js_k}$ is $\sigma^2_{s_k}$, $P_{{\rm out},js_k}\!=\!1\!-\!\exp\bigg(-\frac{2^{\frac{r_{s_k}}{W_j}}-1}{\gamma_{s_k}\sigma^2_{s_k}}\bigg)$ for a Rayleigh fading channel. Let $\overline{P}_{{\rm out},js_k}\!=\!1-P_{{\rm out},js_k}$\footnote{Throughout the paper $\overline{z}\!=\!1-z$.}
 be the probability of the complement event $\overline{O}_{j,s_k}$. This probability of {\bf correct} packet reception is therefore given by
\begin{equation}
\overline{P}_{{\rm out},js_k}\!=\!\exp\bigg(-\frac{2^{\frac{b}{TW_j\left(1-\frac{\tau}{T}\right)}}-1}{\gamma_{s_k}\sigma^2_{s_k}}\bigg)
\label{corrprob}
\end{equation}
Note that the virtual bands are of unity outage probability because the available bandwidth is zero.

The packet correct reception probability of user $p_j$ transmitting to its respective receiver is given by a similar formula as in (\ref{corrprob}) with the respective primary parameters as follows:
\begin{equation}
\overline{P}_{{\rm out},p_j}\!=\!\exp\bigg(-\frac{2^{\frac{b}{TW_j}}-1}{\gamma_{p_j}\sigma^2_{p_j}}\bigg)
\label{corrprob2}
\end{equation}
\section{Stability Analysis of the System $\mathcal{S}$}\label{system_S}
A fundamental performance measure of a buffered communication network is the stability of the queues. Stability can be defined rigorously as follows.
Denote by $Q^{\left(t\right)}$ the length of queue $Q$ at the beginning of time slot $t$. Queue $Q$ is said to be stable if \cite{sadek,tvtjournal}
\begin{equation}\label{stabilityeqn}
    \lim_{x \rightarrow \infty  } \lim_{t \rightarrow \infty  } {\rm Pr}\{Q^{\left(t\right)}<x\}\!=\!1
\end{equation}

In a multiqueue system, the system is stable when {\bf all} queues are stable. We can apply Loynes' theorem to check the stability of a queue \cite{sadek}. This theorem states that if the arrival process and the service process of a queue are strictly stationary, and the average service rate is greater than the average arrival rate of the queue, then the queue is stable. If the average service rate is lower than the average arrival rate, then the queue is unstable.

According to the adopted late-arrival model, the queue $Q_\nu$ evolves as follows:
\begin{equation}
    Q_\nu^{t\!+\!1}\!=\!(Q_\nu^t-\mathbb{D}^t_\nu)^+ \!+\!A^t_\nu,
\end{equation}
where $\mathbb{D}_\nu^t$ is the number of departures from queue $Q_\nu$ at time slot $t$, $A_\nu^t$ is the number of arrivals at $Q_\nu$ at time slot $t$, and $(\zeta)^+$ denotes $\max\{\zeta,0\}$.

 The queue of PU $p_j$ is stable when $\lambda_{p_j} < \mu_{p_j}$. The mean service rate of PU $p_j$ is given by
\begin{equation}
\mu_{p_j}\!=\! \overline{P}_{{\rm out},{p_j}}, \,\ \forall \ j\!\in\mathcal{B}
\label{eqn4}
\end{equation}

  A packet at the queue head of user $s_k$ is served if the band $B_{j}$ assigned to $s_k$ is available and the channel to its respective receiver is not in outage. Define $\mu_{jk}\!=\!\pi_{j} \overline{P}_{{\rm out},{{j}{s_k}}}$, which is the average service rate when band $B_{j}$ is allocated to user $s_k$.  Accordingly, the mean service rate, $\mu_{s_k}$, of user $s_k$ is given by
\begin{equation}
\mu_{s_k}\!=\! \sum_{j\in \mathcal{B}} \ \sum_{\Pi_n\in \mathcal{M}_{jk}}  q(\Pi_n) \mu_{jk}
\label{eqn6}
\end{equation}
Using (\ref{constraints_omega2}), we can write
\begin{equation}
\mu_{s_k}\!=\! \sum_{j\in\mathcal{B}} \omega_{{j}k} \mu_{jk}
\label{omeg_for}
\end{equation}
The expression in (\ref{omeg_for}) can be interpreted as follows: The $k$th SU is served if it is assigned to the primary band $j$, which occurs with probability $\omega_{{j}k}$, while this band is free/available and the associated channel to the $k$th SU respective receiver is not in outage.

The stability region is characterized by the closure of rates $(\lambda_{s_1},\lambda_{s_2},\dots,\lambda_{s_{M_s}})$. One method to characterize this closure is to solve a constrained
optimization problem to find the maximum feasible $\lambda_{s_k}$ corresponding
to each feasible $\lambda_{s_\ell}$, $\ell \neq k$, with all the system queues being stable~\cite{sadek,tvtjournal}. Specifically, for fixed $\lambda_{s_\ell}$, for all  $\ell \neq k$, the maximum stable-throughput region is obtained via solving the following optimization problem:

\begin{equation}
\begin{split}\label{eqn:opt1}
\underset{q(\Pi_n) \geq 0}{\max.} & \ \ \lambda_{s_k} =  \sum_{j\in \mathcal{B}} \ \sum_{\Pi_n\in \mathcal{M}_{jk}}  q(\Pi_n) \mu_{jk}, \\
\text{s.t. } & \sum_{n=1}^{|\mathcal{M}|} q(\Pi_n) = 1, \\
& \lambda_{s_\ell} \le  \sum_{j\in \mathcal{B}} \ \sum_{\Pi_n\in \mathcal{M}_{j\ell}}  q(\Pi_n) \mu_{j\ell}, \forall \ell \neq k.
    \end{split}
    \end{equation}
    The optimization problem in \eqref{eqn:opt1} is a linear program and can be solved using any standard linear programming technique. However, the total number of variables is $|\mathcal{M}|$ which grows very quickly with $M_p$ and $M_s$ according to~\eqref{eqn:cardinality}.

In order to decrease the total number of optimization variables, we use an equivalent optimization problem in terms of $\omega_{jk}$ instead of $q(\Pi_n)$. Defining matrix $\Omega$ such that its $jk$ element is $\omega_{jk}$ and using (\ref{omeg_for}), the optimization problem can be rewritten as follows:
\begin{equation}
\begin{split}\label{eqn:opt2}
\underset{\Omega}{\max.} & \ \ \lambda_{s_k} = \sum_{j\in \mathcal{B}} \omega_{jk} \mu_{jk}\\
\text{s.t. } & 0 \le \omega_{jh} \ \forall j,h ,\\
& \sum_{j=1}^{M_p} \omega_{jh} \le 1 \ \forall h, \ \ \ \sum_{h=1}^{M_s} \omega_{jh} \le 1 \ \forall j, \\
& \lambda_{s_\ell} \le \sum_{j \in \mathcal{B}} \omega_{j\ell}\ \mu_{j \ell } \ \forall \ell \neq k
\end{split}
\end{equation}
  \noindent where $h,\ell\in \{1,2,\dots,M_s\}$. The optimization problem in~\eqref{eqn:opt2} is still a linear program, which can be solved efficiently. It has a total number of variables $M_s \times M_p \ll |\mathcal{M}|$ which is much less than the total number of variables in~\eqref{eqn:opt1}.

\begin{remark}
After solving the optimization problem in~\eqref{eqn:opt2} and obtaining the optimal value of $\Omega$, the operation of the system (see Algorithm~\ref{alg:S}) requires the values of $q(\Pi_n)$, which can be obtained from $\Omega$ using Birkhoff algorithm (see~\cite{chang2000birkhoff,li2001enhanced,chang2002load,peng2006quick} and references therein).
\end{remark}

\begin{remark}
The Birkhoff algorithm is applied on square doubly stochastic matrices.\footnote{A doubly stochastic matrix (also called bistochastic), is a matrix $A\!=\!(a_{jk})$ of nonnegative real numbers and each of its rows and columns sums to unity, i.e., $\underset{j}{\sum} a_{jk}\!=\!\underset{k}{\sum} a_{jk}\!=\!1$.} Therefore, to enable its application in our system, if $M_s>M_p$, we assume that there are virtual bands of zero bandwidth to which $M_s-M_p$ users are assigned. Similarly, if $M_p>M_s$, we assume that there are virtual SUs with zero-arrival rate and unity outage probability.
\end{remark}

\begin{remark}
The optimization problem and the associated optimal solution are functions of only long term statistics of the system such as channel variances, average arrival rates of the SUs, outage probabilities of the links, and probability of the bands being empty or nonempty. There are no dependencies on instantaneous values such as CSI. Thus, the optimization problem can be solved off-line and the corresponding optimal parameters can be used for a long duration of the network life-time. Therefore, once the optimization problem is solved at a central fusion (or a controller),  the controller can supply long sequences of assignment patterns to each user to be used for long operational time. If any of the average parameters change, the controller solves the problem again with the new parameters and feeds the users with the new assignments. Thus, the operation becomes a matter of long term system tuning, which eliminates the need to worry about signalling overhead typically associated with centralized dynamic optimization problems.
\end{remark}

\begin{remark}
The optimal solution of the optimization problems is not unique, in general. However, any of the optimal solutions will provide the same stability
region as they achieve the same rates for users.
\end{remark}

 \begin{proposition}
 The stable-throughput region of system $\mathcal{S}$ is a \textbf{convex} polyhedron.
 \end{proposition}
\begin{proof}
From (\ref{eqn:opt1}) and (\ref{eqn:opt2}), and since the mean service rate of user $s_k$ is affine function of $\omega_{jk}$ and $q(\Pi_n)$ for all $j,k$ and $n$, the stability region, which is the intersection of the constraints, is a convex set, i.e., a polyhedron. Thus, the set of rate tuples $(\lambda_{s_1},\lambda_{s_2},\lambda_{s_3},\dots,,\lambda_{s_{M_s}})$ is convex.
\end{proof}
The stability region being a convex polyhedron
corresponds to a regime in which when one of the SU increases its rate, the other users' maximum supportable rates decrease
linearly. Another interpretation of the convexity of the stability region is
that when the stability region is convex then higher sum rates
can be achieved \cite{naware2005stability}. Also, since the stability region is convex, if two rate pairs are stable then the line segment connecting those two rate pairs is also
composed of stable rate pairs \cite{naware2005stability}.

  \subsection{The Case of Two SUs and Two Primary Bands}
    In this subsection, we move our attention to the case of two SUs and two PUs (two bands) to obtain some insights and analytical results for the stability region. Since $M_s\!=\!M_p\!=\!2$ and from (\ref{constraints_omega1}) and (\ref{constraints_omega}), $\omega_{12}\!=\!\omega_{21}$.\footnote{Since the SUs are assigned different bands in each time slot, the probability of assigning user $s_1$ to band $2$ is equal to the probability of assigning user $s_2$ to band $1$.} The stability region is characterized by the closure of rate pairs $(\lambda_{s_1},\lambda_{s_2})$. The optimization problem is stated as:
    \begin{equation}
\begin{split}
     & \underset{0\le \epsilon\le 1}{\max.} \,\,\,\,\,\,\,\,\,\,\ \epsilon \mu_{12}\!+\! \big(1-\epsilon\big)\mu_{22} \\& \,\,\,\ {\rm s.t.}  \,\,\,\,\ \lambda_{s_1}\!\le \!  \epsilon \mu_{21} \!+\! \big(1\!-\!\epsilon\big) \mu_{11}
     \label{opt_two}
    \end{split}
    \end{equation}
    where $\epsilon=\omega_{12}\!=\!\omega_{21}$ is the probability that user $s_2$ is assigned to band $1$ (or user $s_1$ is assigned to band $2$). The optimization problem can be rearranged as follows:
    \begin{equation}
    \label{opt1001}
\begin{split}
     & \underset{0\le \epsilon\le 1}{\max.} \,\,\,\,\,\,\,\ \epsilon \big(\mu_{12}-\mu_{22}\big)\\ &{\rm s.t.}  \,\,\,\,\,\,\,\ \lambda_{s_1}-\mu_{11}\!\le \! \epsilon \big(\mu_{21}- \mu_{11}\big)
    \end{split}
    \end{equation}
    \begin{proposition}\label{prop1}
      For any network with $M_s=2$ SUs and $M_p~=~2$ orthogonal primary bands, the stability region of system $\mathcal{S}$, $\mathcal{R}(\mathcal{S})$, is given by
               \begin{equation}\label{2233344}
   \mathcal{R}(\mathcal{S})\!=\!\bigg\{(\lambda_{s_1},\lambda_{s_2}):\lambda_{s_2} <  \epsilon^* \mu_{12}\!+\! \bigg(1-\epsilon^*\bigg)\mu_{22}\bigg\}
\end{equation}
   where $\epsilon^*$ denotes the optimal value of $\epsilon$ and is a function of $\lambda_{s_1}$. This value depends on $\mu_{jk}$ for all $j,k\in\{1,2\}$ and $\lambda_{s_1}$. Specifically,
    \begin{itemize}
   \item  If $\mu_{12}>\mu_{22}$, $\mu_{21}< \mu_{11}$ and $\lambda_{s_1}-\mu_{11}<0$, the optimal value is $\epsilon^*\!=\!\min\bigg(\frac{ \lambda_{s_1}-\mu_{11}}{ \mu_{21}- \mu_{11}},1\bigg)$.
          \item  If $\mu_{12}>\mu_{22}$, $\mu_{21}\ge \mu_{11}$ and $\lambda_{s_1} \le \mu_{21}$, the optimal value is $\epsilon^*\!=\!1$.
      \item   If $\mu_{12}<\mu_{22}$ and $\mu_{21}> \mu_{11}$, the optimal value is $\epsilon^*\!=\!\max\bigg(\frac{ \lambda_{s_1}-\mu_{11}}{ \mu_{21}- \mu_{11}},0\bigg)$.
    \item  If $\mu_{12}<\mu_{22}$, $\mu_{21}< \mu_{11}$ and $\lambda_{s_1}\le\mu_{11}$, the optimal value is $\epsilon^*\!=\!0$.
         \item If $\mu_{12}\!=\!\mu_{22}$, the optimization problem becomes a feasibility problem. The optimal solution is a set of $\epsilon^*$ that satisfies the constraints. Note that the SUs can use any of the feasible $\epsilon$ in their operation as any of the points belonging to the optimal set provides the same maximum throughput (maximum objective function).
               \item  If $\mu_{21}< \mu_{11}$ and $\lambda_{s_1}>\mu_{11}$; or $\mu_{21}> \mu_{11}$ and $\lambda_{s_1} > \mu_{21}$, the problem is \textbf{infeasible}.
         \end{itemize}
             \end{proposition}
            \begin{proof}
The first item is explained as follows: If $\mu_{12}>\mu_{22}$, the objective function $\epsilon \big(\mu_{12}-\mu_{22}\big)$ is positive. Hence, the maximum is attained when $\epsilon$ is set to its highest feasible value. If $\mu_{21}< \mu_{11}$ and $\lambda_{s_1}-\mu_{11}<0$, the highest feasible $\epsilon$, from the constraints $ \lambda_{s_1}-\mu_{11}\!\le \! \epsilon (\mu_{21}- \mu_{11})$ and $0\le\epsilon\le 1$, is $\min\bigg(\frac{ \lambda_{s_1}-\mu_{11}}{ \mu_{21}- \mu_{11}},1\bigg)$. Therefore, the optimal $\epsilon$ is $\epsilon^*\!=\!\min\bigg(\frac{ \lambda_{s_1}-\mu_{11}}{ \mu_{21}- \mu_{11}},1\bigg)$.

The second item can be explained as follows:
 If $\mu_{12}>\mu_{22}$, the objective function $\epsilon \big(\mu_{12}~-~\mu_{22}\big)$ is positive. Hence, the maximum is attained when $\epsilon$ is set to its highest feasible value. If $\mu_{21}\ge \mu_{11}$ and $\lambda_{s_1} \le \mu_{21}$, the highest feasible $\epsilon$, from the constraints $ \epsilon\ge \kappa=\underbrace{(\lambda_{s_1}-\mu_{11})}_{\le 0}/ \underbrace{(\mu_{21}- \mu_{11})}_{\ge0}$, where $\kappa\le0$, and $0\le\epsilon\le 1$, is $1$. Hence, the optimal $\epsilon$ is $\epsilon^*=1$.
The other items can be obtained in a similar fashion.
\end{proof}
The stability region of system $\mathcal{S}$ in case of two users and two bands is depicted in Fig.\ \ref{exact_S}.

\begin{figure}
\center
  \includegraphics[width=1\columnwidth]{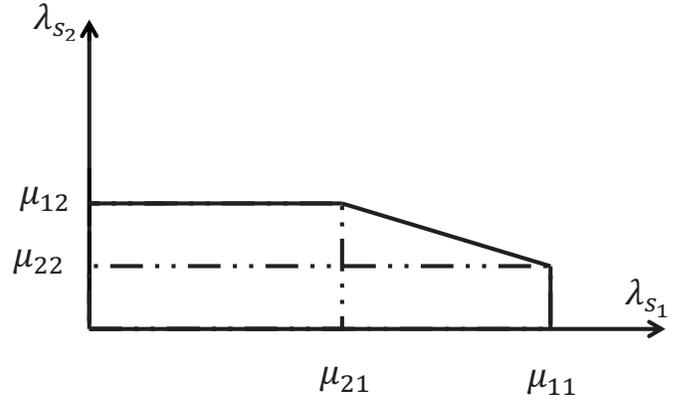}\\
  \caption{Stability of system $\mathcal{S}$. The figure is generated under the assumption that $\mu_{12}>\mu_{22}$ and  $\mu_{11}>\mu_{21}$.}\label{exact_S}
\end{figure}

From the solution, we note that as the rate of user $s_1$ increases, i.e., $\lambda_{s_1}$ increases, the optimal solution is to directly assign user $s_1$ to the band which gives better average throughput for this user. More specifically, if $\mu_{11}>\mu_{21}$, it is more likely to assign user $s_1$ to band $1$ most of the operational time. On the other hand, if $\mu_{11}<\mu_{21}$, user $s_1$ will be assigned to band $2$ most of the operational time. This is motivated by the necessity of stability of user $s_1$ which is maintained by the increase of the service rate of its queue. Let us assume $\mu_{12}>\mu_{22}$ and  $\mu_{11}>\mu_{21}$. At the edge of stability, for $0\le \lambda_{s_2}\le \mu_{22}$, the SU $s_1$ will be assigned to the band with highest $\mu_{j1}$, $j\in\{1,2\}$, i.e., $\mu_{11}$, for all the time, i.e., with probability $1$. Ditto for SU $s_2$. The maximum stable-throughput of $s_2$ for $0\le \lambda_{s_1}\le \mu_{21}$ is $\mu_{12}$. This fact is shown in Fig.\ \ref{exact_S}. We can precisely say that the assignment in those cases is deterministic where the user with low arrival rate is assigned to the band which provides a service rate that merely maintains its stability.

\subsection{The Case of Multiple SUs and One Primary Band}
In this subsection, we investigate the case when only one primary band is available and the other bands are either never idle, i.e., always busy due to the instability of the primary queues assigned to them, i.e., $\lambda>\mu$, or non existing, i.e., only one PU exists in the network. Without loss of generality, we assume that the first band can be empty with certain probability $\pi_1$.

The optimization problem that provides the closure can be written as follows:
    \begin{equation}
\begin{split}
     & \underset{\omega_{1h} \forall h }{\max.} \,\,\,\,\,\,\,\,\,\,\ \omega_{1k} \mu_{1k}\\& \,\,\,\ {\rm s.t.}  \,\,\,\,\ \lambda_{s_\ell}\!\le \! \omega_{1\ell}\mu_{1\ell} \forall \ell\ne k,\,\ 0\le \omega_{1h}\le 1 \forall h \in\{1,2,\dots,M_s\}
    \end{split}
    \end{equation}
The optimization problem is a linear program and can be readily solved. The optimal value of $\omega_{1k}$ is given by
    \begin{equation}
\begin{split}
 \omega_{1h}^*\!=\!1-\frac{\lambda_{s_h}}{\mu_{1h}}
    \end{split}
    \end{equation}
with $\lambda_{s_h}\le \mu_{1h}$. The stability region is given by
         \begin{equation}\label{22333}
   \mathcal{R}(\mathcal{S})\!=\!\bigg\{(\lambda_{s_1},\lambda_{s_2},\dots,\lambda_{s_{M_s}}):\sum_{h=1}^{M_s}\frac{\lambda_{s_h}}{\mu_{1h}}<1\bigg\}
\end{equation}
As is obvious, $\mathcal{R}(\mathcal{S})$ is affine set; hence, convex.
\subsection{The Case of Symmetric SUs}
When the SUs have the same arrival rate, $\lambda_{s_k}=\lambda_s$, and the SUs' channel parameters are equal and therefore all channels outage probabilities are equal for all SUs, the SUs are said to be symmetric. Hence, $ \mu_{{j} k}\!=\!g_{j}$ and $\omega_{jk}\!=\!\theta_j$ for all $k$. In this case, the constraint $\sum_{k\!=\!1}^{M_s} \omega_{jk}\le 1 \ \forall j$ is converted to an upper bound on the feasible value of $\theta_j$. That is, $\theta_j\le 1/M_s$. The optimization problem (\ref{eqn:opt2}) can be rewritten as:
    \begin{equation}
\begin{split}
     &  \underset{\Omega}{\max.}\,\,\ \lambda_{s}\!=\! \sum_{j=1}^{M_p} \theta_j \ g_{j}\\& \,\, {\rm s.t.} \ 0\le \theta_j\!\le\! \frac{1}{M_s}, \ j\!\in\! \{1,2,\dots,M_p\}, \ \sum_{j\in\mathcal{B}}\theta_j=1
    \end{split}
    \end{equation}
    This problem is linear and its exact solution is straightforward. Let us assume without loss of generality that $g_1 ~\ge ~g_2 ~\ge~ g_3~\ge~ \dots~\ge~ g_{M_p}$. To maximize the objective function, we choose
 \begin{equation}
 \begin{split}
    \theta^*_{j} = \begin{cases}
 \frac{1}{M_s},& 0<j\le j_{\max}\\
0,  & j_{\max}<j\le M_p
\end{cases},
\end{split}
\end{equation}
where $j_{\max}\!=\!\min\{M_p,M_s\}$ and $\theta^*_j$ denotes the optimal value of $\theta_j$. If we have $M_p<M_s$, each user is assigned to every one of the $M_p$ bands for $1/M_s$ of the time, and to the null band (zero-bandwidth band) for $1-M_p/M_s$ of the time.

 Based on the optimal solution, we can conclude the following remark. In case of symmetric SUs, the SUs share the {\it best} $\min\{M_s,M_p\}$ primary bands equally likely. The best $\min\{M_s,M_p\}$ primary bands are the $\min\{M_s,M_p\}$ bands with highest $g_{j}$. The optimal solution has the following intuitive explanation. If $M_p\le M_s$, all bands will have one of the SUs each time slot. If $M_p> M_s$, the SUs will select $\min\{M_p,M_s\}\!=\!M_s$ of the primary bands in the channel allocation process. Those bands should be the best in terms of $g_{j}$ because they will provide the highest mean service rates for the secondary queues. If there are two or more bands with the same $g_{j}$, the users share $\min\{M_s,M_p\}$ of the bands even with equal $g_{j}$.
 The maximum secondary mean arrival rate in case of symmetric SUs is then given by
       \begin{equation}
\begin{split}
      \lambda^{\max}_{s}\!=\! \sum_{j=1}^{M_p} \theta^*_{{j}} \ g_{j}
    \end{split}
    \end{equation}
    The stability region is then given by
           \begin{equation}\label{2233344xxxx}
   \mathcal{R}(\mathcal{S})\!=\!\bigg\{(\lambda_{s_1},\lambda_{s_2}, \dots, \lambda_{s_{M_s}}):\lambda_{s_k} <  \sum_{j=1}^{M_p} \theta^*_{{j}} \ g_{j} \forall k\bigg\}
\end{equation}
\subsection{The Case of Symmetric Primary Bands}
Under symmetric bands, the mean arrival rates of the primary queues are equal, i.e., $\lambda_{p_j}\!=\!\lambda_p$ for all $j$, and the channel parameters of all bands are equal. Furthermore, the assigned bandwidth to each primary band is equal, i.e., $W_j\!=\!W$ for all $j$, and $\mu_{ik}\!=\!\mu_{jk}\!=\!\beta_k$ for all $i,j \in\mathcal{B}$ and $k\in\{1,2,\dots,M_s\}$. In this case, the assignments of users will not change the throughput. Specifically, each user gets the same service rate at each band. If $M_p \ge M_s$, the SUs are assigned all the time to any $M_s$ of the $M_p$ primary bands. The mean service rate of each user is fixed over bands and is given by
   \begin{equation}
\begin{split}
    \mu_{s_k}\!=\!\beta_k, \ k\in\{1,2,\dots,M_s\}
    \end{split}
    \end{equation}
    Applying Loynes theorem, the stability region is characterized by
       \begin{equation}
\begin{split}
\mathcal{R}(\mathcal{S})\!=\!\Bigg\{(\lambda_{s_1},\lambda_{s_2},\dots,\lambda_{s_{M_s}}):\lambda_{s_{k}}\!<\! \beta_{k} \forall k\Bigg\}
    \end{split}
    \end{equation}
with $M_p \ge M_s$. This region is a convex \textbf{orthotope} (hyper-rectangle) region.

    If $M_p < M_s$, we need to solve the optimization problem to find the rates' closure. The optimization problem is a linear program. First, we should note that the probability of assigning user $s_k$ to any of the available bands is equal. That is, $\omega_{ik}\!=\!\omega_{jk}\!=\!\eta_k$ for all $i,j \in \mathcal{B}$ and $k\in\{1,2,\dots,M_s\}$. Second, the constraint (\ref{constraints_omega1}) holds to equality. Finally, $\sum_{j=1}^{M_p} \omega_{{j}k}\!=\!\sum_{j=1}^{M_p} \eta_k\!=\!M_p\eta_k$. Substituting by the equality constraint $\sum_{h\!=\!1}^{M_s} \eta_h\!=\! 1\rightarrow \eta_k\!=\!1-\sum_{\substack{{h\!=\!1}\\{ h\ne k}}}^{M_s}\eta_h$ into the objective function, after straightforward simplifications, the optimization problem can be rewritten as follows:
        \begin{equation}
\begin{split}
     &  \underset{0\le \eta_h\le1,\ \forall h }{\min.}\,\,\ \sum_{\substack{{h\!=\!1}\\{ h\ne k}}}^{M_s}\eta_h, \,\, {\rm s.t.} \,\ \eta_h \le\frac{1}{M_p} \ \forall h,\, \ \eta_\ell\ge \frac{\lambda_{s_\ell}}{\beta_{\ell}  M_p}\! \ \forall  \ell \neq k
    \end{split}
    \end{equation}
    Since each term of the sum $\sum_{\substack{{h\!=\!1}\\{ h\ne k}}}^{M_s}\eta_h $ is positive, the minimum of the objective function is attained when the lower constraint of $\eta_h$ holds to equality. That is, $\eta_\ell^*\!=\!\frac{\lambda_{s_\ell}}{\beta_{\ell}  M_p}$ and $\lambda_{s_{\ell}}\le \beta_{\ell}$.\footnote{This condition is obtained from the constraint $\eta_h\le 1/M_p$ which maintains the feasibility of the problem.} Hence, the stability region is characterized by
   \begin{equation}
\begin{split}
\mathcal{R}(\mathcal{S})\!=\!\Bigg\{\!(\lambda_{s_1},\lambda_{s_2},\dots,\lambda_{s_{M_s}})\!:\sum_{k=1}^{M_s} \frac{\lambda_{s_{k}}}{\beta_{k} }\!<\! M_p \bigcap \lambda_{s_{k}}\!<\! \beta_{k} \forall k\Bigg\}
    \end{split}
    \end{equation}
with $M_p < M_s$. Since the stability region when $M_p < M_s$ is the intersection of two affine sets (convex sets), hence it is convex polyhedron.
\subsection{The Case of Symmetric SUs and Symmetric Primary Bands}
  Due to symmetry of bands, $\beta_1\!=\!\beta_2\!=\!\dots\!=\!\beta_{M_p}=\beta$. In this case, each SU is assigned to any of the primary bands with probability $\frac{1}{M_s}$ and each ceases operation (assigned to a null band) with probability $1- \min\{\frac{M_p}{M_s},1\}$. The probability of getting one of the primary bands is $\min\{ \frac{M_p}{M_s},1\}$, where this probability becomes unity in case of $M_p\ge M_s$. Hence, the mean service rate of any of the SUs is $ \min\{\frac{M_p}{M_s},1\}\beta$. If the number of bands is greater than or equal to the number of SUs, each user is assigned to one of the bands all the time. Hence, the mean service rate is characterized by the complement of the channel outage and the band availability (note that due to symmetry, all bands have the same availability probability). That is, the mean service rate of any SU is $\beta$. Combining all cases, the optimal assignment probability is $\omega_{jk}^*\!=\!\frac{1}{M_s}$, where $j\in \mathcal{M}_{\min\{M_p,M_s\}}$ and $\mathcal{M}_{\min\{M_p,M_s\}}\subseteq \mathcal{B}$ is any subset of the $M_p$ bands with cardinality $\min\{M_p,M_s\}$.\footnote{$\mathcal{M}_{\min\{M_p,M_s\}}$ can be any subset of bands with cardinality $\min\{M_p,M_s\}$. Hence, without loss of generality, we can assume that $\mathcal{M}_{\min\{M_p,M_s\}}\!=\!\{1,2,3,\dots,\min\{M_p,M_s\}\}$.} Hence, the maximum stable-throughput is characterized by
   \begin{equation}
\begin{split}
 \lambda^{\max}_s\!=\!\min\Bigg\{\frac{M_p}{M_s},1\Bigg\} \beta
 \label{ffgo}
    \end{split}
    \end{equation}
    Based on the optimal throughput of users, we can get the following conclusions. The throughput increases linearly with the increase of the number of bands, $M_p$, and decreases linearly with the number of SUs, $M_s$. Once the number of bands exceeds (or at least equals to) the number of users, i.e., $M_p\ge M_s$, the secondary achievable throughput becomes totally independent of the number of users and bands. Furthermore, the maximum stable-throughput decreases linearly with the arrival rate of the primary queues $\lambda_p$. That is, $\lambda^{\max}_{s} \propto \beta \propto (1-\lambda_p/\overline{P}_{{\rm out},p})$, where $\overline{P}_{{\rm out},p}$ is the complement of the channel outage between any PU and its respective receiver.

    The stability region of the secondary network in case of symmetric SUs and symmetric bands is given by
        The stability region is then given by
           \begin{equation}\label{2233344xxxxxxx}
   \mathcal{R}(\mathcal{S})\!=\!\bigg\{(\lambda_{s_1},\lambda_{s_2}, \dots, \lambda_{s_{M_s}}):\lambda_{s_k} <  \min\Bigg\{\frac{M_p}{M_s},1\Bigg\} \beta \ \forall k\bigg\}
\end{equation}
    \section{Random Allocation: System $\mathcal{\hat{S}}$}\label{comparisonsection}
    In this section, we consider the first system that we compare to the proposed system, which we refer to as {\it random selection of bands}. This system, denoted by $\mathcal{\hat{S}}$, needs less coordination and cooperation between SUs. Each SU chooses (selects) a primary band randomly at the beginning of the time slot. The probability that user $s_k$ chooses band $B_j$ is $\Gamma_{jk}$. It is clear that these probabilities satisfy the constraint
\begin{equation}
\sum_{j=1}^{M_p}\Gamma_{jk}\le 1,\, \forall k \in\{1,\dots,M_s\}
\end{equation}

It is possible in system $\mathcal{\hat{S}}$ that a band is left unassigned or that several SUs are competing on the same band. In this system, packet loss occurs due to collisions, when two or more users select the same band, as well as due to channel outages. The total number of assignment of SUs to bands is given by
\begin{equation}
|\mathcal{M}_{\hat {\mathcal{S}}}|\!=\!M_p^{M_s}
\end{equation}

 $\mathcal{\hat{S}}$ is less complex than $\mathcal{S}$ because it does not need coordination between the secondary terminals, while in $\mathcal{S}$ coordination is required to guarantee that one and only one user is given a specific band. Nevertheless, the complexity of obtaining the optimal assignments probability in  $\mathcal{\hat{S}}$ is much higher than system $\mathcal{S}$ because the optimization problem of system $\mathcal{\hat{S}}$ is nonconvex and the total number of optimization parameters is $M_p^{M_s}\gg M_p\times{M_s}$ for $M_p>2$.

 The access probabilities are obtained at a control unit (such as one of the SUs). After that the control unit supplies each user with the access/selection probability associated to each band. Upon having the selection probabilities, every time slot each user locally chooses one of the bands using the obtained probabilities. The randomness and distributed manner came from the fact that each user chooses one of the bands locally and without any coordination or cooperation. Accordingly, the possibility of collisions is high.
  We summarize MAC algorithm of system $\mathcal{\hat{S}}$ as shown in Algorithm~\ref{alg:S hat}.
\begin{algorithm} {}
\caption{$\mathcal{\hat{S}}$--MAC}
\begin{algorithmic} \label{alg:S hat}

\WHILE{TRUE}
\STATE{{\it Assignment}:}
\STATE{user $s_k$ selects band $B_{j}$ w.p. $\Gamma_{jk}$}

\STATE{{\it Primary}:}
\FOR{$\forall p_j$}
\IF{$Q_{p_j}$ not empty}
\STATE{transmit packet at head of $Q_{p_j}$}
\ENDIF
\ENDFOR

\FOR{$\forall p_j$}
\IF{ACK received}
\STATE{remove packet at head of $Q_{p_j}$}
\ENDIF
\ENDFOR

\STATE

\STATE{{\it Secondary}:}
\FOR{$\forall s_k$}
\IF{$Q_{s_k}$ is nonempty}
\STATE{sense $B_{m_k}$ for duration $\tau$}
\IF{$B_{m_k}$ idle}
\STATE{transmit packet at head of $Q_{s_k}$}
\ENDIF
\ENDIF
\ENDFOR

\FOR{$\forall s_k$}
\IF{ACK received}
\STATE{remove packet at head of $Q_{s_k}$}
\ENDIF
\ENDFOR

\ENDWHILE

\end{algorithmic}
\end{algorithm}

The mean service rate of the $j${th} PU is similar in systems $\mathcal{S}$ and $\mathcal{\hat{S}}$. We investigate now the service rate for the SUs. User $s_k$, when assigned to band $B_j$, succeeds in its transmission with probability $\overline{P}_{{\rm out},js_k}$ if the PU operating on $B_j$ has no packets to send and if all secondary terminals contending on the same band have empty queues. Recalling that the band assignment is represented by the $M_s$-tuple $(m_1,m_2,\dots,m_k=j,\dots,m_{M_s})$, the mean service rate of user $s_k$ is thus given by
\begin{equation}\label{genwoc}
\begin{split}
\mu_{s_k}\!=\!&\sum_{m_1\!=\!1}^{M_p}\sum_{m_2\!=\!1}^{M_p}...\sum_{m_{M_s}\!=\!1}^{M_p}\Bigg[\Gamma_{m_11}\Gamma_{m_22}...\Gamma_{m_{M_s}M_s}\ \mu_{jk}\\ & \,\,\,\,\,\,\,\,\,\,\,\,\,\,\,\,\,\,\,\,\,\,\,\,\,\,\,\,\,\,\,\,\,\,\,\,\,\,\,\,\,\,\,\,\,\,\,\,\,\,\,\,\,\,\,\,\,\ \times \,{\rm Pr}\Bigg\{ \bigcap_{ \substack{v\in\{1,2,\dots,M_s\} \\ v \neq k \\ m_v \!=\! j}} Q_{s_v} \!=\!0 \Bigg\} \Bigg]
\end{split}
\end{equation}
where the sums in (\ref{genwoc}) are over all possible assignments for every SU.

Due to the complexity of this system and the interaction of queues, we can only study the case of two SUs and one or two primary bands. To analyze the stability of the system's queues, we
resort to a stochastic dominance approach\footnote{It must be noted that stochastic dominance can be used to find the exact stability region only for the case where the assumption of saturation of one queue results in an independent queue system (dominant system) \cite{sadek}, which is true for the case of two SUs and one or two primary bands} \cite{sadek}, where one or set of the nodes is assumed to be saturated while the other nodes operate as they would in the original system. Analyzing the stability
of interacting queues is a difficult problem that has been
addressed for ALOHA systems initially. Characterizing the stable
throughput region for interacting queues is still an open problem \cite{sadek}.

  \subsection{Two SUs and Two Bands}
In this subsection, we focus on the case of two SUs and two PUs (two bands).
 At the beginning of the time slot, the PUs send the packet at the head of their queues. Each SU chooses a band with some probability independent of the other users. If the band is sensed to be idle, the SUs transmit the packet at the head of their queues. The mean service rates of the PUs are given by
\begin{equation}
\mu_{p_1}\!=\! \overline{P}_{{\rm out},{{p_1}}}, \
\mu_{p_2}\!=\! \overline{P}_{{\rm out},{{p_2}}}
\label{eqn50}
\end{equation}
The mean service rates of the SUs are given by
\begin{equation}
\begin{split}
\mu_{s_1}&\!=\!  {\rm Pr}\bigg\{Q_{s_2}\ne0\bigg\}\bigg[ \Gamma_{11}\Gamma_{22}  \mu_{11}  \!+\! \Gamma_{21}\Gamma_{12} \mu_{21}\bigg]\\& \,\,\,\,\,\,\,\,\,\,\,\,\,\,\,\,\,\,\,\,\,\,\,\,\,\,\,\,\,\,\,\,\,\,\,\,\,\,\ \!+\! \bigg[ \Gamma_{11} \mu_{11}  \!+\! \Gamma_{21} \mu_{21}\bigg]{\rm Pr}\bigg\{Q_{s_2}\!=\!0\bigg\}
\label{eqn60}
\end{split}
\end{equation}
\begin{equation}
\begin{split}
\mu_{s_2}&\!=\!  {\rm Pr}\bigg\{Q_{s_1}\!\ne\!0\bigg\}\bigg[ \Gamma_{12}\Gamma_{21}  \mu_{12} \!+\! \Gamma_{22}\Gamma_{11} \mu_{22}\bigg]\!\\&\,\,\,\,\,\,\,\,\,\,\,\,\,\,\,\,\,\,\,\,\,\,\,\,\,\,\ \!+\! \bigg[ \Gamma_{12} \mu_{12}  \!+\! \Gamma_{22} \mu_{22}\bigg]{\rm Pr}\bigg\{Q_{s_1}\!=\!0\bigg\}
\label{eqn70}
\end{split}
\end{equation}
Since the queues are interacting with each other, we resort to the idea of the dominant systems, where the analysis assumes that one of the nodes sends dummy packets when its queue is empty and all the other nodes behave exactly as they would in the original system. We construct two dominant systems and take the union over both of them to obtain the stability of the original system.
\subsubsection{First dominant system}
In the first dominant system, denoted by $\hat{\mathcal{S}}_1$, the queue of user $s_1$ sends dummy packets when it is empty and the other queues behave exactly as they would in the original system. The mean service rate of the SU $s_2$ is given by

\begin{equation}
\begin{split}
\mu_{s_2}&\!=\!   \Gamma_{12}\Gamma_{21}  \mu_{12} \!+\! \Gamma_{22}\Gamma_{11} \mu_{22}
\label{eqn700}
\end{split}
\end{equation}
The probability that the queue of the SU $s_2$ is empty is given by
\begin{equation}
\begin{split}
{\rm Pr}\bigg\{Q_{s_2}\!=\!0\bigg\}\!=\!1-\frac{\lambda_{s_2}}{\mu_{s_2}}
\label{eqn700}
\end{split}
\end{equation}
Therefore, the mean service rate of the SU $s_1$ is given by
\begin{equation}
\begin{split}
\mu_{s_1}&\!=\!  \frac{\lambda_{s_2}}{\mu_{s_2}} \bigg[ \Gamma_{11}\Gamma_{22}  \mu_{11}  \!+\! \Gamma_{21}\Gamma_{12} \mu_{21}\bigg]\\& \,\,\,\,\,\,\,\,\,\,\,\,\,\,\,\,\,\,\,\,\,\,\,\,\,\,\,\,\ \!+\! \bigg[ \Gamma_{11} \mu_{11}  \!+\! \Gamma_{21} \mu_{21}\bigg]\bigg(1\!-\!\frac{\lambda_{s_2}}{\mu_{s_2}}\bigg)
\label{eqn600}
\end{split}
\end{equation}
Based on the construction of the dominant system $\hat{\mathcal{S}}_1$, it can be noted that the lengths of the queues of the dominant system are
never less than those of the original system, provided they are
both initialized identically. This is because, in the dominant
system, node $s_1$ transmits dummy packets even if it does not have
any packets of its own; hence, prevents $s_2$ from transmitting its packets without collisions (or definite packet loss) when $s_1$ chooses the same band. Note that $s_1$ interferes with $s_2$ in all
cases that it would in the original system. Therefore, given that
$\lambda_{s_2}<   \Gamma_{12}\Gamma_{21}  \mu_{12}  + \Gamma_{22}\Gamma_{11} \mu_{22}$, if for some $\lambda_{s_1}$ the queue at
$s_1$ is stable in the dominant system, then the corresponding
queue in the original system must be stable; conversely, if for
some $\lambda_{s_1}$ in the dominant system, the node $s_1$ saturates, then
it will not transmit dummy packets, and as long as $s_1$ has
a packet to transmit, the behavior of the dominant system is
identical to that of the original system. Therefore, we
can conclude that the original system and the dominant system
are \textbf{indistinguishable} at the boundary points. The portion of
the stable-throughput region $\mathcal{R}(\hat{\mathcal{S}}_1)$ which is based on $\hat{\mathcal{S}}_1$ is obtained via solving a constrained
optimization problem to find the maximum feasible $\lambda_{s_1}$ corresponding
to each feasible $\lambda_{s_2}$ as $0\le \Gamma_{11},\Gamma_{12},\Gamma_{21},\Gamma_{22}\le 1$ under the constraints that $\Gamma_{21}\!+\!\Gamma_{11}\!=\!1$ and $\Gamma_{12}\!+\!\Gamma_{22}\!=\!1$. For a fixed $\lambda_{s_2}$, the maximum stable arrival rate for the secondary queue $s_1$ is given by solving the following optimization problem:

\begin{equation}
\begin{split}
\label{0900}
     & \underset{\substack{ {\Gamma_{11},\Gamma_{12}}\\{\Gamma_{21},\Gamma_{22}}}}{\max}\,\,\,\,\ \lambda_{s_1}\!=\!\frac{\lambda_{s_2}}{\mu_{s_2}} \bigg[ \Gamma_{11}\Gamma_{22}  \mu_{11} \!+\! \Gamma_{21}\Gamma_{12} \mu_{21}\bigg]\\&\,\,\,\,\,\,\,\,\,\,\,\,\,\,\,\,\,\,\,\,\,\,\,\,\,\,\,\,\,\,\,\,\,\,\ \!+\! \bigg[ \Gamma_{11} \mu_{11}   \!+\! \Gamma_{21} \mu_{21}\bigg]\bigg(1\!-\!\frac{\lambda_{s_2}}{\mu_{s_2}}\bigg)\\& \,\,\,\,\,\ {\rm s.t.}\,\,\,\,\ \Gamma_{21}\!+\!\Gamma_{11}\!=\!1,\,\ \Gamma_{12}\!+\!\Gamma_{22}\!=\!1\\& \,\,\,\,\,\,\,\,\,\,\,\,\,\,\,\,\,\,\,\,\,\,\,\  \lambda_{s_2}\!\le \! \mu_{s_2}\!=\!\Gamma_{12}\Gamma_{21}  \mu_{12}  \!+\! \Gamma_{22}\Gamma_{11} \mu_{22}
    \end{split}
    \end{equation}
    The optimization problem is nonconvex and can be solved numerically using a two dimensional grid search over $\Gamma_{21}$ and $\Gamma_{12}$ or $\Gamma_{22}$; or $\!\Gamma_{11}$ and $\Gamma_{12}$ or $\Gamma_{22}$, and using the linear constraints, $\Gamma_{21}\!+\!\Gamma_{11}\!=\!1$ and $\Gamma_{12}\!+\!\Gamma_{22}\!=\!1$, to obtain the other parameters.

    {\it \underline{Solution}:} We propose the following simple solution, which converts the problem to a linear program by fixing one of the optimization parameters. Substituting by the equality constraints, we get the optimization problem (\ref{dfgfggg}) at the top of the following page. For a fixed (given) $\Gamma_{12}$, the optimization problem is a linear fractional program on $\Gamma_{22}$, which can be converted to a linear program as explained in \cite[page 151]{boyed}. In our case, we have only one optimization variable for a fixed $\Gamma_{21}$. Therefore, the problem can be readily solved.
     The optimization problem for a fixed $\Gamma_{21}$ is given by


        \begin{figure*}[!t]
\normalsize
\setcounter{equation}{46}
\begin{equation}\label{dfgfggg}
\begin{split}
     & \underset{\substack{ {\Gamma_{21},\Gamma_{22}}}}{\max.}\,\,\,\,\ \frac{\lambda_{s_2}\bigg[ \overline{\Gamma_{21}}\Gamma_{22}  \mu_{11} \!+\! \Gamma_{21}\overline{\Gamma_{22}} \mu_{21}\bigg] \!+\! \bigg[ \overline{\Gamma_{21}} \mu_{11}   \!+\! \Gamma_{21} \mu_{21}\bigg]\Bigg(\Gamma_{21}  \mu_{12}  \!+\! \Gamma_{22}\Big[\overline{\Gamma_{21}} \mu_{22}-\Gamma_{21}  \mu_{12}\Big]\!-\!{\lambda_{s_2}}\Bigg)}{\Gamma_{21}  \mu_{12}  \!+\! \Gamma_{22}\Big[\overline{\Gamma_{21}} \mu_{22}-\Gamma_{21}  \mu_{12}\Big]}\\& \,\,\,\,\,\ {\rm s.t.}\,\,\,\,\  \,\,\,\,\,\,\,\,\,\,\,\,\,\,\,\,\,\,\,\,\,\,\,\  \lambda_{s_2}\!\le \! \!\Gamma_{21}  \mu_{12}  \!+\! \Gamma_{22}\Big[\overline{\Gamma_{21}} \mu_{22}-\Gamma_{21}  \mu_{12}\Big]
 \end{split}
\end{equation}
\hrulefill
\end{figure*}
        \begin{equation}
        \label{dfgfggg1}
\begin{split}
     & \underset{\substack{ {\Gamma_{22}}}}{\max.}\,\,\,\,\ \frac{\Gamma_{22}\bigg[ (1-\Gamma_{21})  \mu_{11}-\Gamma_{21} \mu_{21}\bigg] \!-\!  (1-\Gamma_{21}) \mu_{11}}{\Gamma_{21}  \mu_{12}  \!+\! \Gamma_{22}[(1-\Gamma_{21}) \mu_{22}-\Gamma_{21}  \mu_{12}]}\\& \,\,\,\,\,\ {\rm s.t.}\,\,\,\,\    \lambda_{s_2}-\!\Gamma_{21}  \mu_{12}  \!\le  \Gamma_{22} [(1-\Gamma_{21}) \mu_{22}-\Gamma_{21}  \mu_{12}]
    \end{split}
    \end{equation}
     The problem can be rewritten as follows:
        \begin{equation}
        \label{dfgfggg2}
\begin{split}
     & \underset{\substack{ {\Gamma_{22}}}}{\max.}\,\,\,\,\ \frac{\Gamma_{22}K_1 \!-\!  K_2}{D  \!+\! \Gamma_{22}C}\\& \,\,\,\,\,\ {\rm s.t.}\,\,\,\,\  \lambda_{s_2}-\!\Gamma_{21}  \mu_{12}  \!\le  \Gamma_{22} C
    \end{split}
    \end{equation}
        where $C=[(1-\Gamma_{21}) \mu_{22}-\Gamma_{21}  \mu_{12}]$, $D=\Gamma_{21}  \mu_{12}$, $K_1= (1-\Gamma_{21})  \mu_{11}-\Gamma_{21} \mu_{21}$ and $K_2=(1-\Gamma_{21}) \mu_{11}$. The solution of optimization problem (\ref{dfgfggg2}) is provided in Appendix A.

  Under the proposed technique, we solve a family of linear-fractional programs parameterized by the $\Gamma_{21}$. The optimal solution is then obtained by taking the union over all these linear-fractional programs. For similar technique to find the optimal solution, the reader is referred to \cite{tvtjournal,ourletter}.

\subsubsection{The second dominant system}
In the second dominant system, $\hat{\mathcal{S}}_2$, the queue of user $s_2$ sends dummy packets when it is empty and the other queues behave exactly as they would in the original system. Consequently,
\begin{equation}
\begin{split}
\mu_{s_1}&\!=\!   \Gamma_{11}\Gamma_{22}  \mu_{11}  \!+\! \Gamma_{21}\Gamma_{12} \mu_{21}
\label{eqn60}
\end{split}
\end{equation}

The probability that the queue of the SU $s_1$ is empty is given by
\begin{equation}
\begin{split}
{\rm Pr}\bigg\{Q_{s_1}\!=\!0\bigg\}\!=\!1-\frac{\lambda_{s_1}}{\mu_{s_1}}
\label{eqn7000}
\end{split}
\end{equation}
The mean service rate of user $s_2$ is then given by
\begin{equation}
\begin{split}
\mu_{s_2}&\!=\!  \frac{\lambda_{s_1}}{\mu_{s_1}}\bigg[ \Gamma_{12}\Gamma_{21}  \mu_{12}  \!+\! \Gamma_{22}\Gamma_{11} \mu_{22}\bigg]\\& \,\,\,\,\,\,\,\,\,\,\,\,\,\,\,\,\,\,\,\,\,\,\,\,\,\,\,\,\ \!+\! \bigg[ \! \Gamma_{12} \mu_{12}  \!+\! \Gamma_{22} \mu_{22}\!\bigg]\bigg(1\!-\!\frac{\lambda_{s_1}}{\mu_{s_1}}\bigg)
\label{eqn70e}
\end{split}
\end{equation}

The stability regions of the original system and $\hat{\mathcal{S}}_2$ are indistinguishable at the boundary points. For a fixed $\lambda_{s_1}$, the maximum stable arrival rate for the secondary queue $s_2$ is given by solving the following optimization problem:
\begin{equation}
\begin{split}
\label{pogo}
     &\underset{ \substack{{\Gamma_{11},\Gamma_{12}}\\{\Gamma_{21},\Gamma_{22}}}}{\max.}\,\,\,\ \lambda_{s_2}\!=\!\frac{\lambda_{s_1}}{\mu_{s_1}}\bigg[ \Gamma_{12}\Gamma_{21}  \mu_{12}   \!+\! \Gamma_{22}\Gamma_{11} \mu_{22}\bigg]\\&\,\,\,\,\,\,\,\,\,\,\,\,\,\,\,\,\,\,\,\,\,\,\,\,\,\ \!+\! \bigg[ \Gamma_{12} \mu_{12}  \!+\! \Gamma_{22} \mu_{22}\bigg]\bigg(1-\frac{\lambda_{s_1}}{\mu_{s_1}}\bigg)\\&\,\,\,\,\ {\rm s.t.} \,\,\,\,\,\,\,\,\,\,\,\,\,\,\,\ \Gamma_{21}\!+\!\Gamma_{11}\!=\!1, \,\ \Gamma_{12}\!+\!\Gamma_{22}\!=\!1\\& \,\,\,\,\,\,\,\,\,\,\,\,\,\,\,\,\,\,\,\,\,\,\,\  \lambda_{s_1}\!\le \!\mu_{s_1}\!=\!  \Gamma_{11}\Gamma_{22}  \mu_{11} \!+\! \Gamma_{21}\Gamma_{12} \mu_{21}
    \end{split}
    \end{equation}
       Similar to the first dominant system optimization problem, (\ref{pogo}) can be readily solved.
       This problem can be solved in a similar fashion to (\ref{dfgfggg2}).

       The maximum stable-throughput region of system $\mathcal{\hat{S}}$ is given by the union over the stability sets of the two dominant systems \cite{sadek}, i.e., $\mathcal{R}(\mathcal{\hat{S}})\!=\!\mathcal{R}(\hat{\mathcal{S}}_1)\bigcup \mathcal{R}(\hat{\mathcal{S}}_2)$.
\subsection{The Case of Two SUs and One Primary Band}
This case can be deduced from the previous case by assuming that $\pi_2\!=\!0$. It can be shown that $\mu_{s_1}$ can be rewritten as

\small \begin{equation}\label{serv11}
\begin{split}
         \setcounter{equation}{54}
\mu_{s_1}&\!=\!  {\rm Pr}\bigg\{Q_{s_2}\ne0\bigg\}  \Gamma_{11}\Gamma_{22} \mu_{11}   \!+\!  \Gamma_{11} \mu_{11} {\rm Pr}\bigg\{Q_{s_2}\!=\!0\bigg\}
\end{split}
\end{equation} \normalsize
\noindent Similarly,
\small \begin{equation}\label{serv12}
\begin{split}
   \mu_{s_2}&\!=\!  {\rm Pr}\bigg\{Q_{s_1}\!\ne\!0\bigg\}\mu_{12} \Gamma_{12}\Gamma_{21}  \!+\!  \Gamma_{12} \mu_{11} {\rm Pr}\bigg\{Q_{s_1}\!=\!0\bigg\}
   \end{split}
\end{equation} \normalsize
\noindent Since the two queues are interacting with each other, we resort to the idea of dominant systems to obtain the stability region.

\subsubsection{First dominant system}
In the first dominant system $\mathcal{\hat{S}}_{1}$, $s_1$ transmits dummy packets when its queue is empty, and $s_2$ behaves exactly as it would in the original system. The mean service rate of $Q_2$ is given by
\small \begin{equation}\label{serv22}
\begin{split}
 \setcounter{equation}{56}
   \mu_{s_2}&\!=\!  \Gamma_{12}\Gamma_{21}  \mu_{12}
   \end{split}
\end{equation} \normalsize
\noindent Since ${\rm Pr}\{Q_{s_2}\!=\!0\}\!=\!1-\frac{\lambda_{s_2}}{\mu_{s_2}}$ with $\mu_{s_2}$ given by (56), $\mu_{s_1}$ can be written as

\small \begin{equation}\label{serv21}
   \mu_{s_1}\!=\!\mu_{11}\bigg[\Gamma_{11}  \Gamma_{22}  \frac{\lambda_{s_2}}{\mu_{s_2}}\!+\! (1-\frac{\lambda_{s_2}}{\mu_{s_2}})   \Gamma_{11} \bigg]
\end{equation} \normalsize
\noindent Using the same argument discussed in the previous Subsections, we find the closure of the rate pairs $(\lambda_{s_1},\lambda_{s_2})$. The optimization problem for a fixed $\lambda_{s_2} \le \mu_{s_2}\le \mu_{12}$ can be formulated as

\small \begin{equation}
\begin{split}
& \underset{\Gamma_{11},\Gamma_{12},\Gamma_{21},\Gamma_{22}}{\max.} \,\,\lambda_{s_1}\!=\!   \mu_{s_1}\!=\!\mu_{11}\bigg[\Gamma_{11}  \Gamma_{22}  \frac{\lambda_{s_2}}{\mu_{s_2}}\!+\! (1-\frac{\lambda_{s_2}}{\mu_{s_2}})   \Gamma_{11} \bigg]\\
&\,\,{\rm s.t.} \,\,\,\,\, \Gamma_{ij} \geq 0, i\!=\!1,2, j\!=\!1,2,\,\,\,\,\Gamma_{11}\!+\!\Gamma_{21}\!=\!1,\,\,\Gamma_{12}\!+\!\Gamma_{22}\!=\!1,\,\ \\ & \,\,\,\,\,\,\,\ \lambda_{s_2} \leq  \mu_{s_2}=\Gamma_{12}\Gamma_{21}  \mu_{12}
\end{split}
\label{optprobx}
\end{equation} \normalsize
Letting $K\!=\!\frac{\lambda_{s_2}}{\mu_{12}}$, and using the equality constraints, (\ref{optprobx}) can be expressed as
\small \begin{equation}
\begin{split}
& \underset{\Gamma_{12},\Gamma_{21}}{\min.}\,\,\,\,\Gamma_{21}\!+\!\frac{K}{\Gamma_{21}},
\,\,{\rm s.t.} \,\,\,\, \frac{K}{\Gamma_{21}}- \Gamma_{12} \leq 0,\,\, 0 \leq \Gamma_{12} \leq 1,\,\, 0 \leq \Gamma_{21} \leq 1
\end{split}
\label{optran}
\end{equation} \normalsize
\noindent Note that $K \leq \Gamma_{12} \Gamma_{21} \leq 1$. It can be shown that problem (\ref{optran}) is a convex optimization problem, which can be solved using the Lagrangian formulation.\footnote{The objective function can be shown to be convex by checking the sign of the eigenvalues of the Hessian matrix. Ditto for the constraint function $\frac{K}{\Gamma_{21}}- \Gamma_{12}$. The other inequality constraints are linear.} The optimal probabilities are
\small \begin{equation}
\Gamma^*_{11}\!=\!1-\min\bigg\{\sqrt{\frac{{\lambda_{s_2}}}{{\mu_{12}}}},1\bigg\},\Gamma^*_{12} \!=\! 1,\,\,\Gamma^*_{21}\!=\!1-\Gamma^{*}_{11},\,\,\Gamma^*_{22}\!=\!0
\label{gamopt1}
\end{equation} \normalsize
The maximum stable-throughput region of the first dominant system $\mathcal{R}(\mathcal{\hat{S}}_{1})$ is given by
\small \begin{equation}\label{2234443}
\mathcal{R}(\mathcal{\hat{S}}_{1})= \!\bigg\{(\lambda_{s_1},\lambda_{s_2}): \sqrt{\frac{\lambda_{s_1}}{\mu_{11}} }\!+\!\sqrt{\frac{\lambda_{s_2}}{\mu_{12}}}<1 \bigg\}
\end{equation} \normalsize
\subsubsection{Second dominant system}
In the second dominant system $\mathcal{\hat{S}}_{2}$, $s_2$ transmits dummy packets when its queue is empty, whereas $s_1$ operates exactly as it would in the original system. Following the analysis of $\mathcal{\hat{S}}_{1}$, we obtain the following results
\small \begin{equation}
\Gamma^*_{11} \!=\! 1, \,\,\Gamma^*_{12}\!=\!1-\min\big\{\sqrt{\frac{{\lambda_{s_1}}}{{ \mu_{11}} }},1\big\},\,\,
\Gamma^*_{21}\!=\!0, \,\,\Gamma^*_{22}\!=\! 1-\Gamma^{*}_{12}
\label{gamopt2}
\end{equation} \normalsize
\noindent The stability region of the second dominant system $\mathcal{R}(\mathcal{\hat{S}}_{2})$ is given by
\small \begin{equation}\label{2233445}
   \mathcal{R}(\mathcal{\hat{S}}_{2})\!=\!\bigg\{(\lambda_{s_1},\lambda_{s_2}): \sqrt{\frac{\lambda_{s_1}}{\mu_{11}}}\!+\!\sqrt{\frac{\lambda_{s_2}}{\mu_{12}}}<1 \bigg\}
\end{equation} \normalsize
\noindent The stability region of system $\mathcal{\hat{S}}$ is $\mathcal{R}(\mathcal{\hat{S}})\!=\!\mathcal{R}(\mathcal{\hat{S}}_1)\bigcup \mathcal{R}(\mathcal{\hat{S}}_2)$. That is,
\small \begin{equation}\label{223344885}
      \mathcal{R}(\mathcal{\hat{S}})=   \mathcal{R}(\mathcal{\hat{S}}_{1})=\mathcal{R}(\mathcal{\hat{S}}_{2})\!=\!\bigg\{(\lambda_{s_1},\lambda_{s_2}): \sqrt{\frac{\lambda_{s_1}}{\mu_{11} }}\!+\!\sqrt{\frac{\lambda_{s_2}}{\mu_{12}}}<1 \bigg\}
\end{equation} \normalsize
We note that the stability region is not convex. This means that an increase in the maximum
rate of one SU implies a disproportionate decrease of the other.

\begin{proposition}\label{pro2}
For any network with $M_s$ SUs and $M_p$ orthogonal primary bands, the stability region of system $\mathcal{S}$, $\mathcal{R}(\mathcal{S})$, contains
that of $\hat{\mathcal{S}}$, $\mathcal{R}(\mathcal{\hat{S}})$. That is, $\mathcal{R}(\mathcal{\hat{S}}) \subseteq  \mathcal{R}(\mathcal{S})$.
\end{proposition}
\begin{proof}
See Appendix B.
\end{proof}
     \section{Fixed Allocation: System $\mathcal{S}^{\left(\rm F\right)}$}\label{deterministic1000}
     In this system, denoted by $\mathcal{S}^{\left(\rm F\right)}$, each SU is assigned to a certain band individually all the time, i.e., every SU is permanently and uniquely assigned one of the primary bands. Hence, this system requires that $M_p\ge M_s$. Using the notation used for system $\mathcal{S}$, let $\Pi_n$ represent a permutation on $(m_1,m_2,m_3,\dots,m_{M_s})$. Also, let $d(\Pi_n)$ denote a mapping function that maps the SU $s_1$ to band $m_1$ and user $s_2$ to band $m_2$ and so on, where $m_k\ne m_\ell, \ \forall k,\ell$.

     The average service rates of the secondary queues are given by
     \begin{equation}
\mu_{s_k}\!=\! \mu_{{m_k} {s_k}}
\label{eqn1000}
\end{equation}
where $\mu_{m_k {s_k}}$ is the mean service rate for SU $s_k$ given that band $B_{m_k}$ is allocated to it and $k\!\in\{1,2,\dots,M_s\}$ and $m_k\!\in\{1,2,\dots,M_p\}$. The stability region for the case $d(\Pi_n)$, $\mathcal{R}\big(d(\Pi_n)\big)$, is given by

   \begin{equation}\label{analy333}
   \begin{split}
  &\mathcal{R}\big(d(\Pi_n)\big)\!\\& \!=\!\bigg\{(\lambda_{s_1},\lambda_{s_2},\dots,\lambda_{s_{M_s}}): 0\!<\!\lambda_{s_k}\!<\! \mu_{m_k {s_k}} \forall k\!\in\!\{1,2,\dots,M_s\}\! \bigg\}
   \end{split}
\end{equation}
    with all assignments of users are \textbf{distinct} to each other, i.e., $m_k\ne m_\ell$ $\forall k,\ell$. The stability region of this system given a certain allocation permutation is an \textbf{orthotope} (hyper-rectangle) region, which is \textbf{convex}.
%

    In case of two SUs and two bands, i.e., $M_s\!=\!M_p\!=\!2$, the stability region of system $\mathcal{S}^{\left(\rm F\right)}$, using the mapping functions $d(2,1)$ and $d(1,2)$, are given by
      \begin{equation}\label{analy34}
   \mathcal{R}\big(d(2,1)\big)\!=\!\bigg\{(\lambda_{s_1},\lambda_{s_2}):\lambda_{s_2} < \mu_{12}, \,\,\ 0<\lambda_{s_1}<  \mu_{21} \bigg\}
\end{equation}
      \begin{equation}\label{analy34}
   \mathcal{R}\big(d(1,2)\big)\!=\!\bigg\{(\lambda_{s_1},\lambda_{s_2}):\lambda_{s_2} < \mu_{11}, \,\,\ 0<\lambda_{s_1}<  \mu_{22} \bigg\}
\end{equation}
 Depicted in Fig. \ref{deterministic}, the two user two band case stability region of the system $\mathcal{S}^{\left(\rm F\right)}$.

\begin{proposition}\label{pro1}
For $M_s$ SUs and $M_p\ge M_s$ bands, the stability regions of system $\mathcal{S}$ and $\mathcal{\hat{S}}$ contain that of a fixed assignment.
\end{proposition}
\begin{proof}
The fixed assignment system is a special case of system $\mathcal{S}$ corresponding to the case where the probability $q(\Pi_n)$ of the assignment of a certain permutation is unity and all the other probabilities are zero. In addition, the fixed assignment system is a special case of system $\hat{\mathcal{S}}$ with  $\Gamma_{jk}$ set to unity when band $B_j$ is allocated to $s_k$ and zero otherwise. Therefore, both systems $\mathcal{S}$ and $\hat{\mathcal{S}}$ are superior to a fixed assignment.
\end{proof}

\vspace{-0.2cm}
\section{Numerical Results and Simulations}\label{numerical}
\vspace{-0.1cm}
We provide here some insightful numerical results for the systems presented in this work. Let $d(m_1,m_2)$ denote the fixed allocation of user $s_1$ to band $m_1$ and user $s_2$ to band $m_2$ in a system with $M_s\!=\!M_p\!=\!2$. Fig.\ \ref{fig:subfig1} provides a comparison between the stability regions of systems $\mathcal{S}$, $\hat{\mathcal{S}}$, $d(1,2)$ and $d(2,1)$. The parameters used to generate the figure are: $\overline{P}_{{\rm out},2{s_1}}\!=\!0.8$, $\overline{P}_{{\rm out},2{s_2}}\!=\!0.9$, $\overline{P}_{{\rm out},1{s_1}}\!=\!0.7$, $\overline{P}_{{\rm out},1{s_2}}\!=\!0.85$, and the bands availability are $\pi_1=\!1\!-\!\frac{\lambda_{p_1}}{\overline{P}_{{\rm out},p_1}}\!=\! 0.25$ and $\pi_2\!=\!1\!-\!\frac{\lambda_{p_2}}{\overline{P}_{{\rm out},p_2}}\!=\!0.875$. From the figure, the advantage of system $\mathcal{S}$ and $\hat{\mathcal{S}}$ over the deterministic assignment is noted. Also, the advantage of $\mathcal{S}$ over all the considered systems is noted. It can be noted that, the performances of all systems are equivalent at low values of $\lambda_{s_1}$ and low values of $\lambda_{s_2}$. This is because the assignment of users at such cases is deterministic (fixed). We can precisely say that the assignment in those cases is deterministic where the user with low arrival rate is assigned to the band which provides a service rate that merely maintains its stability. The fixed assignment is optimal when $\lambda_{s_1}\le \mu_{11}=0.175$ packets/slot and when $\lambda_{s_2}\le \mu_{12}=0.2125$ packets/slot. We note that for $\lambda_{s_1}> \mu_{11}$, the stable-throughput of user $s_2$ in system $\hat{\mathcal{S}}$ starts to degrade significantly. This is because the arrival rate to user $s_1$ increases and the possibility of collisions increases due to the selection of the same band; hence, packets loss increases and data retransmission is needed. Therefore, the achievable throughput for user $s_2$ is low. This does not happen in case of system $\mathcal{S}$ because collisions never occur.

\begin{figure}
\center
  \includegraphics[width=1\columnwidth]{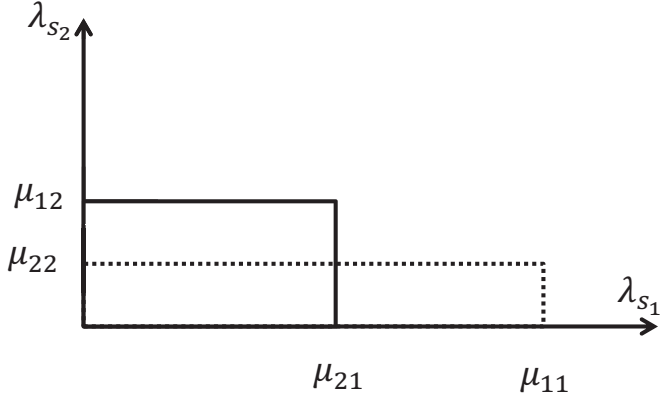}\\
  \caption{Stability of the deterministic system in case of two SUs and two bands. The two deterministic assignment possibilities are depicted in the figure, the solid one for $\mathcal{R}(d(2,1))$ and the dotted one is $\mathcal{R}(d(1,2))$. The figure is generated under the assumption that $\mu_{12}>\mu_{22}$ and  $\mu_{11}>\mu_{21}$.}\label{deterministic}
\end{figure}

Fig. \ref{fig:subfig2} shows the stability region of system $\mathcal{S}$ in case of $M_s=4$ and $M_p=5$. The figure reveals the impact of increasing the mean arrival rate of users $s_3$ and $s_4$ on the stability region of users $s_1$ and $s_2$. As shown in the figure, the increase in the mean arrival rates of users $s_3$ and $s_4$ reduces the stability region of users $s_1$ and $s_2$. The parameters used to generate the figure are depicted in the figure's caption and Table \ref{table2}. Fig. \ref{fig:subfig3} presents the optimal assignments probabilities for system $\mathcal{S}$ for the given parameters in the figure's caption. The parameters used to generate the figure are: $M_s\!=\!M_p\!=\!3$, $\lambda_{s_3}\!=\!\lambda_{s_4}\!=\!0.35$ packets per time slot and the first three rows and columns of users $s_1$, $s_2$ and $s_3$ in Table \ref{table2}. It can be noted that as the mean arrival rate of the second user, $s_2$, increases, $q^*(1,3,2)$ and $q^*(2,3,1)$ increase as well, which can be interpreted as the fraction of time slots that user $s_2$ is allocated to the third band. This is because the third band provides the highest $\mu_{jk}$ for user $s_2$, i.e., $\mu_{32}>\mu_{j2}$ for $j\!\in\!\{1,2\}$, and user $s_2$ needs to increase its service rate to maintain its queue stability. Similarly, as the mean service rate of user $s_1$ increases, the probabilities $q^*(3,2,1)$ and $q^*(3,1,2)$ increase for the same reason mentioned before for user $s_2$. Note that the summation of $q^*(1,3,2)$ and $q^*(2,3,1)$ results in $\omega_{32}$, and the summation of $q^*(3,2,1)$ and $q^*(3,1,2)$ results in $\omega_{31}$. From the figures, we note the convexity of the stability region of $\mathcal{S}$ and its envelope. This actually verifies our observations about the convexity of the stability region of $\mathcal{S}$.
\begin{figure}
	\centering
		\includegraphics[width=1\columnwidth]{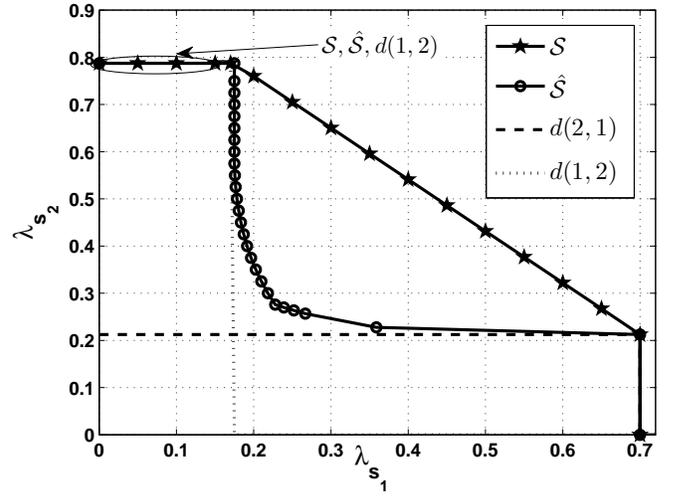}
	\caption{Stability regions of the considered systems.}
	\label{fig:subfig1}
\end{figure}

\section{conclusions and Future Work}\label{conc}
We have proposed a band allocation scheme for buffered cognitive radio users in presence of orthogonal licensed primary bands each of which assigned to a PU.
The cognitive radio users are allocated to bands based on their queue stability requirements.
We have proved the advantage of the proposed scheme over some well-known schemes.

\begin{table*}
\begin{center}
\caption{The complement of channels outage for the secondary nodes and the bands availability of the primary bands used to generate Figs. \ref{fig:subfig2} and \ref{fig:subfig3}.}
\label{table2}
\begin{tabular}{|c|c|c|c|c|}
    \hline\hline
\hbox{User $s_1$}& \hbox{User $s_2$}&\hbox{User $s_3$}& \hbox{User $s_4$} & \hbox{Band Availability}\\[5pt]\hline
    $\overline{P}_{{\rm out},1{s_1}}\!=\!0.6$ &$\overline{P}_{{\rm out},1{s_2}}\!=\!0.7$& $\overline{P}_{{\rm out},1{s_3}}\!=\!0.6$ &$\overline{P}_{{\rm out},1{s_4}}\!=\!0.7$&$\pi_1\!=\!1-\frac{\lambda_{p_1}}{\overline{P}_{{\rm out},p_1}}\!=\!0.45$ \\
      $\overline{P}_{{\rm out},2{s_1}}\!=\!0.8$ &$\overline{P}_{{\rm out},2{s_2}}\!=\!0.6$ &$\overline{P}_{{\rm out},2{s_3}}\!=\!0.8$ &$\overline{P}_{{\rm out},2{s_4}}\!=\!0.5$ &$\pi_2\!=\!1-\frac{\lambda_{p_2}}{\overline{P}_{{\rm out},p_2}}\!=\!0.2$\\
     $ \overline{P}_{{\rm out},3{s_1}}\!=\!0.7$ &$\overline{P}_{{\rm out},3{s_2}}\!=\!0.8 $ &  $ \overline{P}_{{\rm out},3{s_3}}\!=\!0.7$ &$\overline{P}_{{\rm out},3{s_4}}\!=\!0.6 $ &$\pi_3\!=\!1-\frac{\lambda_{p_3}}{\overline{P}_{{\rm out},p_3}}\!=\!0.6$\\
     $\overline{P}_{{\rm out},4{s_1}}\!=\!0.85 $ &$ \overline{P}_{{\rm out},4{s_2}}\!=\!0.9 $& $\overline{P}_{{\rm out},4{s_3}}\!=\!0.5 $ &$ \overline{P}_{{\rm out},4{s_4}}\!=\!0.95 $&$\pi_4\!=\!1-\frac{\lambda_{p_4}}{\overline{P}_{{\rm out},p_4}}\!=\!0.4$ \\
     $\overline{P}_{{\rm out},5{s_1}}\!=\!0.9 $ &$ \overline{P}_{{\rm out},5{s_2}}\!=\!0.95 $& $\overline{P}_{{\rm out},5{s_3}}\!=\!0.95 $ &$ \overline{P}_{{\rm out},5{s_4}}\!=\!0.95 $&$\pi_5\!=\!1-\frac{\lambda_{p_5}}{\overline{P}_{{\rm out},p_5}}\!=\!0.6$  \\[5pt]\hline
\end{tabular}
\end{center}
\end{table*}

\begin{figure}
	\centering
		\includegraphics[width=1\columnwidth]{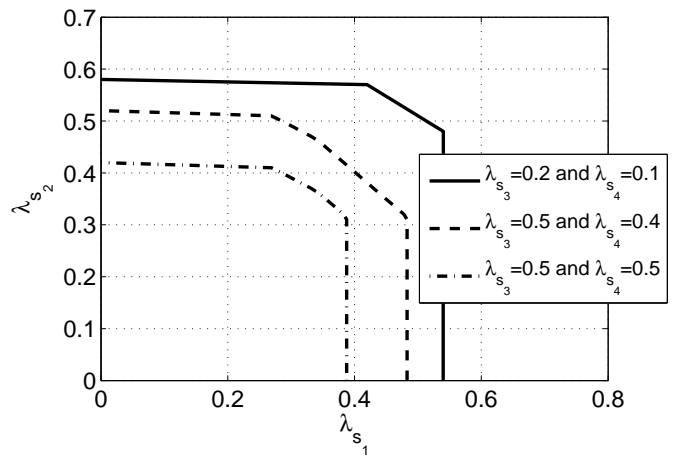}
	\caption{Stability region of system $\mathcal{S}$. The parameters used to generate the figure are: $M_s=4$ and $M_p=5$ and Table \ref{table2}.}
	\label{fig:subfig2}
\end{figure}
\begin{figure}
	\centering
		\includegraphics[width=1\columnwidth]{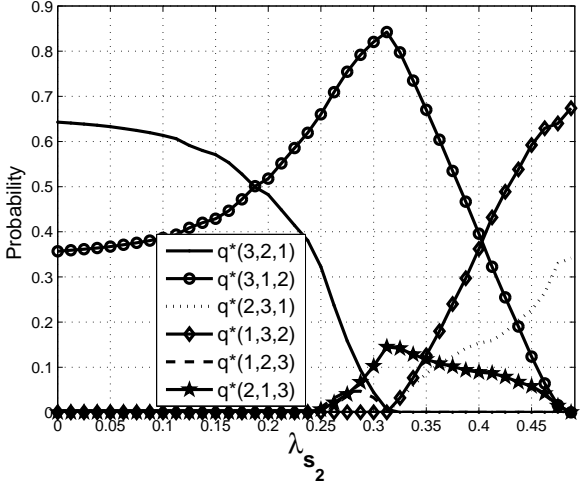}
	\caption{The optimal SUs' allocation probabilities for system $\mathcal{S}$ in case of $M_s\!=\!M_p\!=\!3$. The parameters used to generate the figure are $\lambda_{s_3}\!=\!\lambda_{s_4}\!=\!0.35$ packets per time slot and the first three rows and the columns of users $s_1$, $s_2$ and $s_3$ in Table \ref{table2}.}
	\label{fig:subfig3}
\end{figure}


 Future research for system $\mathcal{S}$ can be directed at one of the following points. 1) Considering systems with multiple assignment within one slot. More specifically, the assignment of users happens multiple time per slot to satisfy all users. The knowledge of the transmit CSI can enhance the system performance and allow bands exchange among users; 2) allowing priority among SUs such that multiple users can be assigned to the same band with different priority in band accessing. The priority of transmission can be established by making the lower priority user sense the higher priority user activity for certain time duration within the slot; or 3) another possible extension is to study the impact of sensing errors on the system's performance. For system $\hat{\mathcal{S}}$, the extension can be directed in terms of 1) adding multipacket reception capabilities to the receiving nodes; or 2) allowing band selection at different time instants per slot followed by sensing duration to avoid perturbing the current transmission~\cite{tvtjournal}.

\section*{Appendix A}
          In this Appendix, we provide the solution of optimization problem (\ref{dfgfggg2}). The first derivative of the objective function of (\ref{dfgfggg2}) with respect to $\Gamma_{22}$ for a fixed $\Gamma_{21}$ is given by

                \begin{equation}
\begin{split}
\frac{\partial}{\partial \Gamma_{22}}\frac{\Gamma_{22} K_1-K_2}{D+C\Gamma_{22}}=  \frac{K_2 C+D K_1}{( C\Gamma_{22}+D)^2}
    \end{split}
    \end{equation}
where
            \begin{equation}
\begin{split}
&C K_2+ DK_1\!=\!(1-\Gamma_{21})^2 \mu_{22} \mu_{11}- \Gamma_{21}^2  \mu_{12} \mu_{21}
    \end{split}
    \end{equation}
    Based on the first derivative, $\lambda_{s_2}-\!\Gamma_{21}  \mu_{12}$, and the value of $C$, the optimal solution of $\Gamma_{22}$, for a fixed $\Gamma_{21}$, is obtained as follows:
    \begin{itemize}
    \item If the derivative is positive, i.e., $C K_2\!+\!D K_1\!>\!0$, the maximum of the objective function is attained when $\Gamma_{22}$ is adjusted to its highest feasible value. Using the constraints, the highest feasible value of $\Gamma_{22}$, which represents the optimal solution of the optimization problem, is obtained as follows:
        \begin{itemize}
        \item If $C>0$ and $\frac{\lambda_{s_2}-\Gamma_{21} \mu_{12} }{C}\le1$, the optimal $\Gamma_{22}$ is $\Gamma^*_{22}=1$.
        \item If $C>0$ and $\frac{\lambda_{s_2}-\Gamma_{21} \mu_{12} }{C}>1$, the problem is infeasible.
        \item If $C<0$, $\lambda_{s_2}-\Gamma_{21}  \mu_{12} < 0$, and $\frac{\lambda_{s_2}-\Gamma_{21} \mu_{12}}{C}>0$, the optimal $\Gamma_{22}$ is $\Gamma^*_{22}=\min\{\frac{\lambda_{s_2}-\Gamma_{21} \mu_{12}}{C},1\}$.
        \item If $C<0$ and $\lambda_{s_2}-\Gamma_{21}  \mu_{12} > 0$, the problem is infeasible.
\end{itemize}
 \item If the derivative is negative, i.e., $C K_2\!+\!D K_1\!<\!0$, the maximum of the objective function is attained when $\Gamma_{22}$ is set to its lowest feasible value. Using the constraints, the lowest feasible value of $\Gamma_{22}$, which represents the optimal solution of the optimization problem, is obtained as follows:
        \begin{itemize}
        \item If $C>0$ and $\frac{\lambda_{s_2}-\Gamma_{21} \mu_{12} }{C}\le1$, the optimal $\Gamma_{22}$ is $\Gamma^*_{22}=\max\{\frac{\lambda_{s_2}-\Gamma_{21} \mu_{12} }{C},0\}$.
                \item If $C>0$ and $\frac{\lambda_{s_2}-\Gamma_{21} \mu_{12} }{C}>1$, the problem is infeasible.
        \item If $C<0$ and $\lambda_{s_2}-\Gamma_{21}  \mu_{12} < 0$, the optimal $\Gamma_{22}$ is $\Gamma^*_{22}=0$.
        \item If $C<0$ and $\lambda_{s_2}-\Gamma_{21}  \mu_{12} > 0$, the problem is infeasible.
\end{itemize}
\end{itemize}
\section*{Appendix B}
In this Appendix, we prove the advantage of system $\mathcal{S}$ over system $\hat{\mathcal{S}}$.
\begin{proof}
We investigate the system with $M_p \geq M_s$ first. Assume the same pattern of queue occupancy in both systems. A packet departs the queue of user $s_k$ if user $s_k$ selects band $B_j$ and all nonempty queue users do not select band $B_j$, band $B_j$ is available, and the channel between user $s_k$ and its destination is not in outage. The mean service rate of user $s_k$ with a nonempty queue is
\begin{equation}
\begin{split}
\mu^{\left(\hat{\mathcal{S}}\right)}_{s_k}&\!=\!  \sum_{j=1}^{M_p} \mu_{jk}\ \Gamma_{{{j}k}} \prod_{\substack{{v\in \mathcal{N}}\\{v\ne k}}}(1-\Gamma_{{j}v})
\end{split}
\label{hatSe}
\end{equation}
\noindent where $\mathcal{N}$ is the set of SUs with nonempty queues. Note that we use the superscript $\hat{\mathcal{S}}$ to make it clear that expression (\ref{hatSe}) is for system $\hat{\mathcal{S}}$.
Using (\ref{omeg_for}) for the service rate of user $s_k$ under system $\mathcal{S}$, and subtracting (\ref{hatSe}) from (\ref{omeg_for}), we get
\begin{equation}
\begin{split}
\mu^{\left(\mathcal{S}\right)}_{s_k}\!-\!\mu^{\left(\hat{\mathcal{S}}\right)}_{s_k}&\!=\! \sum_{j=1}^{M_p}  \omega_{j k} \ \mu_{jk} \!-\!\sum_{j=1}^{M_p} \mu_{jk} \ \Gamma_{{jk}} \prod_{\substack{{v \in \mathcal{N}}\\{v\ne k}}}(1\!-\!\Gamma_{jv})\\ & =\! \sum_{j=1}^{M_p}\mu_{jk} \bigg(\omega_{j k} - \Gamma_{{{j}k}} \prod_{\substack{{v \in \mathcal{N}}\\{v\ne k}}}(1-\Gamma_{{j}v})\bigg)
\end{split}
\end{equation}
Note that $\sum_{j=1}^{M_p}\Gamma_{{{j}k}} \prod_{\substack{{v \in \mathcal{N}}\\{v\ne k}}}(1-\Gamma_{jv})$ represents the probability of one user being assigned a certain band with all other users with nonempty queues being assigned to another band. This configuration is a subset of all possible users' assignments which additionally include a situation with two or more users with nonempty queues assigned to a band and the rest of users assigned to another band. This means that the sum given by $\sum_{j=1}^{M_p}\Gamma_{{{j}k}} \prod_{\substack{{v \in \mathcal{N}}\\{v\ne k}}}(1-\Gamma_{{j}v})$ is less than or equal to $1$. Since $\sum_{j=1}^{M_p}\omega_{j k}\!=\!1$, we can always find $\omega_{j k} \ge \Gamma_{{{j}k}} \prod_{\substack{{v \in \mathcal{N}}\\{v\ne k}}}(1-\Gamma_{{j}v})$.

For completeness, it should be shown that if $\omega_{j k} = \Gamma_{{{j}k}} \prod_{\substack{{v \in \mathcal{N}}\\{v\ne k}}}(1-\Gamma_{{j}v})$, the result satisfies the constraint that $\sum_{k\in \mathcal{N}}\omega_{jk}\le1$. That is, $\sum_{k\in \mathcal{N}}\omega_{jk}=\sum_{k\in \mathcal{N}}\Gamma_{{{j}k}} \prod_{\substack{{v \in \mathcal{N}}\\{v\ne k}}}(1-\Gamma_{{j}v})\le 1$. The probability of an SU individually assigned to band $j$ is $\sum_{k\in \mathcal{N}}\Gamma_{{jk}} \prod_{\substack{{v \in \mathcal{N}}\\{v\ne k}}}(1-\Gamma_{jv})\le 1$. Hence, $\sum_{k\in \mathcal{N}}\omega_{j k}$, when $\omega_{j k} = \Gamma_{{{j}k}} \prod_{\substack{{v \in \mathcal{N}}\\{v\ne k}}}(1-\Gamma_{{j}v})$, is less than or equal to unity. This completes the first part of the proof.

 Now, if $M_p\le M_s$, this case can be seen as a system with $M_p\!=\!M_s$ with $M_s- M_p$ zero-bandwidth bands. Thus, we can infer that $\mathcal{R}(\mathcal{S})$ contains $\mathcal{R}(\hat{\mathcal{S}})$ in all cases. This completes the proof.
\end{proof}
\balance
\bibliographystyle{IEEEtran}
\bibliography{IEEEabrv,bandsbib}

\end{document}